\documentclass{lmcs}
\pdfoutput=1

\usepackage{lastpage}
\lmcsdoi{18}{1}{29}
\lmcsheading{}{\pageref{LastPage}}{}{}%
{Apr.~07,~2021}{Feb.~03,~2022}{}

\usepackage[utf8]{inputenc}

\usepackage[UKenglish]{babel}

\usepackage{algorithm, clrscode3e}

\usepackage{amssymb, mathtools}

\usepackage{stmaryrd}

\usepackage{tikz}
\usetikzlibrary{automata, backgrounds, fit, positioning, shapes}
\tikzset{every state/.style={font=\normalsize, inner sep=1pt, minimum size=16pt}}
\tikzset{tr/.style={->, shorten >=1pt, font=\footnotesize, inner sep=3pt, text depth=0pt, anchor=base, auto}}
\tikzset{initial text=, every initial by arrow/.style={shorten >=1pt}}
\tikzset{automaton/.style={x=1.5cm, y=1.5cm}}
\tikzset{poset/.style={x=1.33cm, y=1.0cm, font=\footnotesize, inner sep=1pt}}
\tikzset{orbit/.style={fill=black!15, rounded corners=4pt, inner sep=3pt}}

\renewcommand*{\id}{\mathsf{id}}
\newcommand*{\up}[1]{\text{\upshape #1}}  
\newcommand*{\mikobrace}[2]{\underbrace{#1}_{\mathclap{\text{\scriptsize #2}}}}

\newcommand*{\aut}{\mathcal{A}}
\newcommand*{\lang}{\mathcal{L}}
\newcommand*{\stlang}[1]{\llbracket #1 \rrbracket}
\newcommand*{\Anchor}{\mathsf{Anc}}
\newcommand*{\anc}[1]{#1^{\mathsf{anc}}}
\newcommand*{\tp}[1]{#1^{\top}}
\newcommand*{\hl}[1]{{\color{red}#1}}

\newcommand*{\atoms}{\mathbb{A}}
\DeclareMathOperator{\Perm}{\mathsf{Perm}}
\DeclareMathOperator{\supp}{\mathsf{supp}}
\DeclareMathOperator{\orb}{\mathsf{orb}}
\DeclareMathOperator{\Orb}{\mathsf{Orb}}
\newcommand*{\Aorb}[1]{{#1}\text{-}\!\orb}

\DeclareMathOperator{\Pow}{\mathcal{P}}
\DeclareMathOperator{\Powfs}{\mathcal{P}_{\mathsf{fs}}}
\DeclareMathOperator{\Powufs}{\mathcal{P}_{\mathsf{ufs}}}
\DeclareMathOperator{\Powfin}{\mathcal{P}_{\mathsf{fin}}}

\newcommand*{\clos}[1]{\left\langle #1 \right\rangle}
\DeclareMathOperator{\Res}{\mathsf{Der}}
\DeclareMathOperator{\JI}{\mathsf{JI}}
\DeclareMathOperator{\JIufs}{\mathsf{JI}_{\mathsf{ufs}}}

\newcommand*{\lstar}{\ensuremath{\mathsf{L^\star}}}
\newcommand*{\nlstar}{\ensuremath{\mathsf{NL^\star}}}
\newcommand*{\nomlstar}{\ensuremath{\nu\mathsf{L^\star}}}
\newcommand*{\nomnlstar}{\ensuremath{\nu\mathsf{NL^\star}}}
\newcommand*{\alstar}{\ensuremath{\mathsf{AL^\star}}}
\newcommand*{\ot}{T}
\newcommand*{\row}{\textsf{row}}
\newcommand*{\Rows}{\textsf{Rows}}
\newcommand*{\RowsUpp}{\Rows^{\uparrow}}

\newenvironment{definition}{\defi}{\enddefi}
\newenvironment{theorem}{\thm}{\endthm}
\newenvironment{proposition}{\prop}{\endprop}
\newenvironment{lemma}{\lem}{\endlem}
\newenvironment{corollary}{\cor}{\endcor}
\newenvironment{remark}{\rem}{\endrem}
\newenvironment{example}{\exa}{\endexa}

\theoremstyle{definition}\newtheorem{construction}[thm]{Construction}
\theoremstyle{plain}\newtheorem{claim}[thm]{Claim}

\newcommand{\ch}[1]{#1}

\keywords{nominal automata, residual automata, derivative language, decidability, closure, exact learning, lattice theory}

\begin{document}

\title[Residuality and Learning for Nondeterministic \ch{Nominal} Automata]{Residuality and Learning for\texorpdfstring{\\}{} Nondeterministic Nominal Automata{\rsuper*}}
\titlecomment{{\lsuper*}This is an extended version of the conference paper~\cite{MoermanS20}. This research has been partially funded by the ERC AdG project 787914 FRAPPANT and EPSRC Standard Grant CLeVer (EP/S028641/1).}

\author[J.~Moerman]{Joshua Moerman\rsuper{a}}
\address{RTWH Aachen University, Germany}
\email{joshua@cs.rwth-aachen.de}

\author[M.~Sammartino]{Matteo Sammartino\rsuper{b}}
\address{Royal Holloway University of London, United Kingdom}
\email{matteo.sammartino@rhul.ac.uk}

\begin{abstract}
We are motivated by the following question: which data languages admit an active learning algorithm?
This question was left open in previous work \ch{by the authors}, and is particularly challenging for languages recognised by nondeterministic automata.
To answer it, we develop the theory of \emph{residual \ch{nominal} automata}, a subclass of nondeterministic \ch{nominal} automata.
We prove that this class has canonical representatives, which can always be constructed via a finite number of observations.
This property enables active learning algorithms, and makes up for the fact that residuality --- a semantic property --- is undecidable for \ch{nominal} automata.
Our construction for canonical residual automata is based on a machine-independent characterisation of residual languages,
for which we develop new results in nominal lattice theory.
Studying residuality in the context of nominal languages is a step towards a better understanding of learnability of automata with some sort of nondeterminism.
\end{abstract}

\maketitle

\section{Introduction}

Formal languages over infinite alphabets have received considerable attention recently.
They include data languages for reasoning about XML databases~\cite{NevenSV04}, trace languages for analysis of programs with resource allocation~\cite{GrigoreDPT13}, and behaviour of programs with data flows~\cite{HowarJV19}.
Typically, these languages are accepted by \emph{register automata}, first introduced in the seminal paper~\cite{KaminskiF94}. \ch{Another appealing model is that of \emph{nominal automata}~\cite{BojanczykKL14}.
While nominal automata are as expressive as register automata,
they enjoy convenient properties.
For example, the deterministic ones admit canonical minimal models, and the theory of formal languages and many textbook algorithms generalise smoothly.}

In this paper, we investigate the properties of so-called \emph{residual} \ch{nominal} automata.
An automaton accepting a language $\lang$ is residual whenever the language of each state is a \emph{derivative} of $\lang$.
In the context of regular languages over finite alphabets, residual finite state automata (RFSAs) are a subclass of nondeterministic finite automata (NFAs) introduced by Denis et al.~\cite{DenisLT02} as a solution to the well-known problem of NFAs \emph{not having} unique minimal representatives.
They show that every regular language admits a unique canonical RFSA, which can be much smaller than the canonical deterministic automaton.

Residual automata play a key role in the context of \emph{exact learning}\footnote{Exact learning is also known as query learning or active (automata) learning~\cite{Angluin87}.}, in which one computes an automaton representation of an unknown language via a finite number of observations.
The defining property of residual automata allows one to (eventually) observe the semantics of each state independently.
In the finite-alphabet setting, residuality underlies the seminal algorithm \lstar{} for learning deterministic automata~\cite{Angluin87} (deterministic automata are always residual), and enables efficient algorithms for learning nondeterministic~\cite{BolligHKL09} and alternating automata~\cite{AngluinEF15, BerndtLLR17}. Residuality has also been studied for learning probabilistic automata~\cite{DenisE08}.
Existence of canonical residual automata is crucial for the convergence of these algorithms.

Our investigation of residuality in the context of data languages is motivated by the question: which data languages admit an exact learning algorithm?
In previous work~\cite{MoermanS0KS17}, we have shown that the \lstar{} algorithm generalises smoothly to data languages, meaning that deterministic \ch{nominal} automata can be learned.
However, the nondeterministic case proved to be significantly more challenging.
In fact, in stark contrast with the finite-alphabet case, nondeterministic \ch{nominal} automata are \emph{strictly more expressive} than deterministic ones, so that residual automata are not just succinct representations of deterministic languages.
As a consequence, our attempt to \ch{generalise the \nlstar{} algorithm for learning nondeterministic finite automata~\cite{BolligHKL09}} to \ch{nominal} automata only partially succeeded: we only proved that it converges for deterministic languages, leaving the nondeterministic case open.
By investigating residual data languages, we are finally able to settle this case.

\begin{figure}[t]
\centering
\begin{tikzpicture}[x=1.5cm, y=1.0cm]
\node (dfa)   at (0, 0)  {\textsc{Deterministic}};
\node (rfsa-) at (0, 1)  {\textsc{Residual}$^{-}$};
\node (rfsa)  at (-1, 2) {\textsc{Residual}};
\node (nfa-)  at (1, 2)  {\textsc{Nondeterministic}$^{-}$};
\node (nfa)   at (0, 3)  {\textsc{Nondeterministic}};
\path[-]
  (dfa)   edge (rfsa-)
  (rfsa-) edge (rfsa)
  (rfsa-) edge (nfa-)
  (rfsa)  edge (nfa)
  (nfa-)  edge (nfa);
\end{tikzpicture}
\caption{Relationship between classes of data languages. Edges are strict inclusions. With~${\cdot}^{-}$ we denote classes where automata are not allowed to \emph{guess} values, i.e., to store symbols in registers without explicitly reading them.}
\label{fig:lang-diag}
\end{figure}

In summary, our contributions are as follows:
\begin{itemize}
\item Section~\ref{sec:examples}: We refine classes of data languages as depicted in Figure~\ref{fig:lang-diag}, by giving separating languages for each class.
\item Section~\ref{sec:crna}: We develop new results of nominal lattice theory, from which we prove the main characterisation theorem (Theorem~\ref{thm:residual-characterisation}).
This provides a machine-independent characterisation of the languages accepted by residual nominal automata and constructs canonical automata which:
a) are minimal in their respective class and unique up to isomorphism; b) can be constructed via a finite number of observations of the language.
We also give an analogous result for non-guessing automata (Theorem~\ref{thm:non-guessing-residuals}).
\item Section~\ref{sec:dec}: We study decidability and closure properties for residual nominal automata.
We prove that, like for nondeterministic nominal automata, equivalence is undecidable.
On the other hand, universality is decidable.
\item Section~\ref{sec:learning}: We settle important open questions about exact learning of data languages. We show that residuality does not imply convergence of existing algorithms, and we give a (modified) \nlstar{}-style algorithm that works precisely for residual languages.
\end{itemize}

This research mirrors that of \emph{residual probabilistic automata} \cite{DenisE08}.
There, too, one has distinct classes of which the deterministic and residual ones admit canonical automata and have an algebraic characterisation.
We believe that our results contribute to a better understanding of learnability of automata with some sort of nondeterminism.

\subsection{Differences with conference version.}
This paper is the extended version of~\cite{MoermanS20}, published at CONCUR'20. Since then we have added:
\begin{itemize}
	\item full proofs for all the results;
	\item a new result stating that equivalence is undecidable;
	\item a greatly expanded Section~\ref{sec:learning}, with the learning algorithm, proofs and bounds;
	\item more elaborate discussion in which we consider alternating and unambiguous nominal automata, as well as other symmetries.
\end{itemize}

\section{Preliminaries}
\label{sec:prelim}

\subsection{Nominal Sets}

Our paper is based on the theory of nominal sets of~\cite{Pitts13}, whose basic notions we now briefly recall.

Let $\atoms = \{ a, b, c, \ldots \}$ be a countably infinite set of \emph{atoms} and let $\Perm(\atoms)$ be the set of \emph{\ch{finite} permutations on $\atoms$}, i.e., the bijective functions $\pi \colon \atoms \to \atoms$ \ch{such that the set $\{a \in \atoms \mid \pi(a) \neq a\}$ is finite}.
\ch{Finite} permutations form a group where the unit is given by the identity function, the inverse by functional inverse, and multiplication by function composition.

A function  ${\cdot} \colon \Perm(\atoms) \times X \to X$ is a \emph{group action} if it satisfies $\pi \cdot ( \pi' \cdot x) = ( \pi \circ \pi') \cdot x$ and $\id \cdot x = x$.
We often omit the ${\cdot}$ and write $\pi x$ instead.
We say that a set of atoms $A \subset \atoms$ \emph{supports} $x \in X$ whenever $\pi x = x$ for all the permutations $\pi$ that fix $A$ pointwise (i.e., $\pi(a) = a$ for all $a \in A$).
A \emph{nominal set} is a set $X$ together with a group action such that each $x \in X$ has a \emph{finite} support (that is, each $x$ is \emph{finitely supported}).
Every element of a nominal set $x \in X$ has a \ch{\emph{least} finite support} which we denote by $\supp(x)$.

Given a nominal set $X$, the group action extends to subsets of $X$.
Let $U \subseteq X$ be a subset, then we define the group action as $\pi \cdot U \coloneqq \{ \pi x \mid x \in U \}$.
A subset $U$ that is supported by the empty set is called \emph{equivariant} and for such $U$ we have $\pi U = U$ for all permutations $\pi$.
This definition extends to relations and functions.
For instance, a function $f \colon X \to Y$ between nominal sets is equivariant whenever $\pi f(x) = f(\pi x)$.

Two elements of a nominal set $x, y \in X$ can be considered equivalent if there is a permutation $\pi$ such that $\pi x = y$.
This defines an equivalence relation and partitions the set into classes.
Such an equivalence class is called an \emph{orbit}.
Concretely, the orbit of $x$ is given by $\orb(x) \coloneqq \left\{ \pi x \mid \pi \in \Perm(\atoms) \right\}$.
A nominal set $X$ is \emph{orbit-finite} whenever it is a finite union of orbits.
More generally, we consider \emph{$A$-orbits}, those are orbits defined over permutations fixing a finite set $A \subset \atoms$, namely $\Aorb{A}(x) \coloneqq \left\{ \pi x \mid \pi \in \Perm(\atoms),\, \forall a \in A \colon \pi(a) = a \right\}$.
We note that ($A$-)orbit-finite sets are ``robust'', in the sense that results assuming equivariance (i.e., $A = \emptyset$) lift to an arbitrary $A$ without much effort (see e.g.,~\cite[Chapter 5]{Pitts13}).
Orbit-finite sets are finitely-representable, hence algorithmically tractable~\cite{BojanczykBKL12}.

\ch{
\begin{example}
Consider the set $D = \{(a,b) \mid a,b \in \atoms\}$, together with the point-wise group action $\pi(a,b) = (\pi(a),\pi(b))$. This forms a nominal set where $\supp(a,b) = \{a,b\}$, for each $(a,b) \in D$. The set $D$ has two orbits:
\[
	\{(a,b) \mid a,b \in A, a\neq b\} \quad\text{and}\quad \{(a,a) \mid a \in \atoms\} \enspace ,
\]
because permutations can never map $(a,b)$ to $(a,a)$ or $(b,b)$ due to bijectivity. It has five $\{a\}$-orbits:
\begin{gather*}
	\{(a,b) \mid b \in \atoms \setminus \{a\} \} \qquad \{(b,a) \mid b \in \atoms \setminus \{a\} \} \qquad \{(b,c) \mid b,c \in \atoms \setminus \{a\}, b \neq c \} \\ 
	\{(a,a)\} \qquad \{(b,b) \mid b \in \atoms \setminus \{a\} \} \enspace .
\end{gather*}
Note that fixing $a$ in permutations has the effect of partitioning each regular orbit into several $\{a\}$-orbits.
\end{example}
}

\subsection{Nominal Power Set}

Given a nominal set $X$, the \emph{nominal power set} is defined as
\[ \Powfs(X) \coloneqq \left\{ U \subseteq X \mid U \text{ is finitely supported} \right\} . \]
Equivalently, one could define $\Powfs(X)$ as the nominal set of all finitely supported functions of type $X \to 2$.
This set is a Boolean algebra of which all the operations (union, intersection, and complement) are equivariant maps~\cite{GabbayLP11}.

We will also use the \emph{finite power set} defined as
\[ \Powfin(X) \coloneqq \left\{ U \subseteq X \mid U \text{ is finite} \right\} . \]
This does not form a Boolean algebra, as it does not have a top element.
However, it is still a join-semilattice.
In Section~\ref{sec:crna} we develop more of the theory of nominal lattices.

\begin{example}
The set $\Powfs(\atoms)$ contains all finite subsets of $\atoms$ as well as all co-finite subsets of $\atoms$.
The support of a finite set $A \subset \atoms$ is simply $A$ itself, and the support of a co-finite set $B$ is given by $\atoms \setminus B$.
On the other hand $\Powfin(\atoms)$ only contains the finite subsets.
We notice that $\Powfin(\atoms) \subsetneq \Powfs(\atoms) \subsetneq \Pow(\atoms)$.
\end{example}

Note that neither power sets preserve orbit-finite sets.
This fact often hinders generalising certain algorithms, such as the subset construction on nondeterministic automata, to the nominal setting.

\subsection{Other Symmetries}\label{sec:symmetries}
So far we have introduced nominal sets based on pure atoms $(\atoms, {=})$, which only allow data values to be compared for equality.
The theory of orbit-finite sets can be parametrised over other data theories.
Another structure often considered is the dense linear order $(\mathbb{Q}, \leq)$, which allow data values to be compared for order.
We refer to~\cite{BojanczykKL14} for other structures with applications to \ch{nominal} automata.

In this paper we restrict to the pure atoms, because of the following useful property, used in Lemma~\ref{lem:join-of-ji}.

\begin{lemma}
Given a poset $(P, \leq)$ where $P$ is a nominal set and $\leq$ is equivariant,
\begin{enumerate}
\item if $x \leq y$ and $x$ and $y$ are in the same orbit, then $x = y$;
\item if $P$ is orbit-finite and non-empty, then it has a minimal element (i.e., some $m \in P$ such that $n \leq m$ implies $n = m$).
\end{enumerate}
\label{lem:minimal-in-poset}
\end{lemma}
\begin{proof}
Ad (1), let $x, y$ in the same orbit.
There is a finite permutation $\pi$ such that $y = \pi x$, which gives $x \leq y = \pi x$.
By equivariance of $\leq$ we get $\pi x \leq \pi y$, which gives
\[ x \leq \pi x \leq \pi^2 x \leq \pi^3 x \leq \cdots \]
Since $\pi$ is a finite permutation we have $\pi^k = \id$ for some $k$ and so $x \leq y \leq \cdots \leq x$, proving $x = y$.

\newcommand*{\localleq}{\preceq}
Ad (2), consider the finite set of orbits $\Orb(P) \coloneqq \{ \orb(x) \mid x \in P \}$ and the induced relation \ch{${\localleq} \subseteq \Orb(P) \times \Orb(P)$} defined by $o_1 \localleq o_2$ if there are $x \in o_1$ and $y \in o_2$ such that $x \leq y$.
The relation $\localleq$ is clearly reflexive. For transitivity, consider $o_1 \localleq o_2 \localleq o_3$ with witnesses $x \leq y$ and $y' \leq z$;
there is a permutation $\pi$ such that $y = \pi y'$, giving $x \leq y = \pi y' \leq \pi z$, which shows $o_1 \leq o_3$.
For antisymmetry, consider $o_1 \localleq o_2$ and $o_2 \localleq o_1$ with witnesses $x \leq y$ and $y' \leq x'$.
Again there is a permutation $\pi$ such that $y = \pi y'$ which gives $x \leq y = \pi y' \leq \pi x'$.
By (1) we have $\pi x' = x$ and hence $x = y$ as required.
We have shown that $(\Orb(P), {\localleq})$ is a finite poset, therefore it has a minimal orbit $o$ in $\Orb(P)$, and any element of $o$ is minimal in $P$.
\end{proof}

The above property does not hold for the ordered atoms $(\mathbb{Q}, \leq)$.
To see this, consider the nominal poset $\mathbb{Q}$ itself (which is a single orbit set) with the given order.
This poset has no minimal element.

\subsection{Nominal Automata}

The theory of nominal automata seamlessly extends classical automata theory by having orbit-finite nominal sets and equivariant functions in place of finite sets and functions.

\begin{definition}
A \emph{(nondeterministic) nominal automaton} $\aut$ consists of an orbit-finite nominal set $\Sigma$, the alphabet, an orbit-finite nominal set of states $Q$, and the following equivariant subsets
\[ \mikobrace{I \subseteq Q}{initial states} \qquad \mikobrace{F \subseteq Q}{final states} \qquad \mikobrace{\delta \subseteq Q \times \Sigma \times Q}{transitions} . \]
\end{definition}

The usual notions of acceptance and language apply.
We denote the language accepted by $\aut$ by $\lang(\aut)$, and the language accepted by a state $q$ by $\stlang{q}$.
Note that the language $\lang(\aut) \in \Powfs(\Sigma^*)$ is equivariant, and that $\stlang{q} \in \Powfs(\Sigma^*)$ need not be equivariant, but it is supported by $\supp(q)$.
\begin{remark}
In most examples we take the alphabet to be $\Sigma = \atoms$, but it can be any orbit-finite nominal set. For instance, $\Sigma = \mathsf{Act} \times \atoms$, where $\mathsf{Act}$ is a finite set of actions, represents actions $\mathsf{act}(x)$ with one parameter $x \in \atoms$ (actions with arity $n$ can be represented via $n$-fold products of $\atoms$).
\end{remark}

We recall the notion of \emph{derivative language}~\cite{DenisLT02}.\footnote{This is sometimes called a \emph{residual language} or \emph{left quotient}. We do not use the term residual language here, because residual language will mean a language accepted by a residual automaton.}
\begin{definition}
Given a language $\lang$ and a word $u \in \Sigma^*$,
we define the \emph{derivative of $\lang$ with respect to $u$} as
\[ u^{-1} \lang \coloneqq \left\{ w \mid uw \in \lang \right\} \]
and the set of all derivatives as
\[ \Res(\lang) \coloneqq \left\{ u^{-1} \lang \mid u \in \Sigma^* \right\} . \]
\end{definition}
These definitions seamlessly extend to the nominal setting. Note that $w^{-1}\lang$ is finitely supported whenever $\lang$ is.

Of special interest are the deterministic, residual, and non-guessing nominal automata, which we introduce next.
\begin{definition}
A nominal automaton $\aut$ is:
\begin{itemize}
\item
\emph{Deterministic} if $I = \{q_0\}$, and for each $q \in Q$ and $a \in \Sigma$ there is a unique $q'$ such that $(q, a, q') \in \delta$.
In this case, the relation is in fact functional $\delta \colon Q \times \Sigma \to Q$.
\item
\emph{Residual} if each state $q \in Q$ accepts a derivative of $\lang(\aut)$, formally: $\stlang{q} = w^{-1}\lang(\aut)$ for some word $w \in \Sigma^*$.
The words $w$ such that $\stlang{q} = w^{-1}\lang(\aut)$ are called \emph{characterising words} for the state $q$.
\item
\emph{Non-guessing} if $\supp(q_0) = \emptyset$, for each $q_0 \in I$, and $\supp(q') \subseteq \supp(q) \cup \supp(a)$, for each $(q, a, q') \in \delta$.
\end{itemize}
\label{def:nom-aut}
\end{definition}
Observe that the transition function of a deterministic automaton preserves supports (i.e., if $C$ supports $(q, a)$ then $C$ also supports $\delta(q, a)$).
Consequently, all deterministic automata are non-guessing. For the sake of succinctness, in the following we drop the qualifier ``nominal'' when referring to these classes of nominal automata.

For many examples, it is useful to define the notion of an anchor.
Given a state $q$, a word $w$ is an \emph{anchor} if $\delta(I, w) = \{ q \}$, that is, the word $w$ leads to $q$ and no other state.
Every anchor for $q$ is also a characterising word for $q$ (but not vice versa).
A state with an anchor is called \emph{anchored}, and we call an automaton \emph{anchored} if all states have anchors.

Finally, we recall the Myhill-Nerode theorem for nominal automata.
\begin{theorem}[{\cite[Theorem~5.2]{BojanczykKL14}}]
\label{thm:my-ne}
Let $\lang$ be a language.
Then $\lang$ is accepted by a deterministic automaton if and only if $\Res(\lang)$ is orbit-finite.
\end{theorem}

\section{Separating languages}\label{sec:examples}

Deterministic, nondeterministic and residual automata have the same expressive power when dealing with finite alphabets.
The situation is more nuanced in the nominal setting.
We now give one language for each class in Figure~\ref{fig:lang-diag}.
For simplicity, we mostly use the one-orbit nominal set of atoms $\atoms$ as alphabet.
These languages separate the different classes, meaning that they belong to the respective class, but not to the classes below or beside it.

For each example language $\lang$, we depict: a nominal automaton recognising $\lang$ (on the left); the set of derivatives $\Res(\lang)$ (on the right). We make explicit the poset structure of $\Res(\lang)$: grey rectangles represent orbits of derivatives, and lines stand for set inclusions (we grey out irrelevant ones). This poset may not be orbit-finite, in which case we depict a small, indicative part. Observing the poset structure of $\Res(\lang)$ explicitly is important for later, where we show that the existence of residual automata depends on it.
We write $aa^{-1}\lang$ to mean $(aa)^{-1}\lang$.
Variables $a,b,\ldots$ are always atoms and $u,w,\ldots$ are always words.

\newcommand{\currlang}{\lang_\up{d}}
\subsection*{Deterministic: First symbol equals last symbol}
Consider the language
\[ \currlang \coloneqq \{ awa \mid a \in \atoms, w \in \atoms^{*} \} . \]
This is accepted by the following deterministic nominal automaton (Figure~\ref{fig:det}).
The automaton is actually infinite-state, but we represent it symbolically using a register-like notation, where we annotate each state with the current value of a register.
Note that the derivatives $a^{-1}\currlang, b^{-1}\currlang, \ldots$ are in the same orbit.
In total $\Res(\currlang)$ has three orbits, which correspond to the three orbits of states in the deterministic automaton.
The derivative $awa^{-1}\currlang$, for example, equals $aa^{-1}\currlang$.

\begin{figure}[h!]
\centering
\begin{tikzpicture}[automaton, baseline=(current bounding box.center)]
\node[initial,state,initial text={$\aut_\up{d} = $}] (q0) at (0,0) {};
\node[state] (q1) at (1,0) {$a$};
\node[state,accepting] (q2) at (2,0) {$a$};

\path[tr]
(q0) edge node {$a$} (q1)
(q1) edge[loop above] node {$\neq a$} (q1)
(q1) edge[bend left] node {$a$} (q2)
(q2) edge[loop above] node {$a$} (q2)
(q2) edge[bend left] node {$\neq a$} (q1);
\end{tikzpicture}\hspace{2cm}
\begin{tikzpicture}[poset, baseline=(current bounding box.center)]
\node (L) at (0, 0.5) {$\currlang$};
\node (aL) at (1, 0) {$a^{-1}\currlang$};
\node (bL) at (2, 0) {$b^{-1}\currlang$};
\node (cL) at (3, 0) {\strut$\cdots$};
\node (aaL) at (1, 1) {$aa^{-1}\currlang$};
\node (bbL) at (2, 1) {$bb^{-1}\currlang$};
\node (ccL) at (3, 1) {\strut$\cdots$};
\draw (aL) -- (aaL) (bL) -- (bbL) (cL) -- (ccL);
\begin{scope}[on background layer]
\node [orbit, fit=(L)] {};
\node [orbit, fit=(aL)(bL)(cL)] {};
\node [orbit, fit=(aaL)(bbL)(ccL)] {};
\end{scope}
\end{tikzpicture}
\caption{A deterministic automaton accepting $\currlang$, and the poset $\Res(\currlang)$.}
\label{fig:det}
\end{figure}

\renewcommand{\currlang}{\lang_\up{ng,r}}
\subsection*{Non-guessing residual: Some atom occurs twice}
The language is
\[ \currlang \coloneqq \{ uavaw \mid u,v,w \in \atoms^{*}, a \in \atoms \} . \]
The poset $\Res(\currlang)$ is not orbit-finite, so by the nominal Myhill-Nerode theorem there is no deterministic automaton accepting $\currlang$.
However, derivatives of the form $ab^{-1}\currlang$ can be written as a union $ab^{-1}\currlang = a^{-1}\currlang \cup b^{-1}\currlang$.
In fact, we only need an orbit-finite set of derivatives to recover $\Res(\currlang)$.
These orbits are highlighted in the diagram on the right (Figure~\ref{fig:ng-res}).
Selecting the ``right'' derivatives is the key idea behind constructing residual automata in Theorem~\ref{thm:residual-characterisation}.

\begin{figure}[h!]
\centering
\begin{tikzpicture}[automaton, baseline=0.1cm]
\node[initial,state,initial text={$\aut_\up{ng,r} = $}] (q0) at (0,0) {};
\node[state] (q1) at (1,0) {$a$};
\node[state,accepting] (q2) at (2,0) {};

\path[tr]
(q0) edge[loop above] node {$\atoms$} (q0)
(q0) edge node {$a$} (q1)
(q1) edge[loop above] node {$\atoms$} (q1)
(q1) edge node {$a$} (q2)
(q2) edge[loop above] node {$\atoms$} (q2);
\end{tikzpicture}
\hspace{2cm}
\begin{tikzpicture}[poset, baseline=(current bounding box.center)]
\node (L) at (0,0) {$\currlang$};
\node (aL) at (-0.75, 1) {$a^{-1}\currlang$};
\node (cL) at (0, 1) {$\cdots$};
\node (bL) at (0.75, 1) {$b^{-1}\currlang$};
\node (abL) at (0, 2) {$ab^{-1}\currlang$};
\node (cdL) at (-1, 2) {$\cdots$};
\node (dcL) at (1, 2) {$\cdots$};
\node (abcL) at (0, 3) {$abc^{-1}\currlang$};
\node (cdeL) at (-1, 3) {$\cdots$};
\node (edcL) at (1, 3) {$\cdots$};
\node (aaL) at (0, 4.5) {$aa^{-1}\currlang$};
\draw (L) -- (aL) (L) -- (bL) (aL) -- (abL) (bL) -- (abL) (abL) -- (abcL);
\draw[dashed] (abcL) -- (aaL);
\begin{scope}[on background layer]
\draw[black!15] (L) -- (cL) (cL) -- (cdL) (cL) -- (dcL) (aL) -- (cdL) (bL) -- (dcL) (cdL) -- (cdeL) (cdL) -- (abcL) (abL) -- (cdeL) (abL) -- (edcL) (dcL) -- (abcL) (dcL) -- (edcL);
\draw[dashed, black!15] (cdeL) -- (aaL) (edcL) -- (aaL);
\node [orbit, draw=red, fit=(L)] {};
\node [orbit, draw=red, fit=(aL)(bL)(cL)] {};
\node [orbit, fit=(abL)(cdL)(dcL)] {};
\node [orbit, fit=(abcL)(cdeL)(edcL)] {};
\node [orbit, draw=red, fit=(aaL)] {};
\end{scope}
\end{tikzpicture}
\caption{A nondeterministic automaton accepting $\currlang$, and the poset $\Res(\currlang)$. The depicted automaton is not residual, but a residual automaton exists by Theorem~\ref{thm:residual-characterisation}.}
\label{fig:ng-res}
\end{figure}

\renewcommand{\currlang}{\lang_\up{n}}
\subsection*{Nondeterministic: Last letter is unique}
We consider the language
\[ \currlang \coloneqq \{ wa \mid a \text{ not in } w \} \cup \{ \epsilon \} . \]
Derivatives $a^{-1}\currlang$ are again unions of smaller languages: $a^{-1}\currlang = \bigcup_{b \neq a} ab^{-1}\currlang$.
However, the poset $\Res(\lang)$ has an infinite descending chain of languages (with an increasing support), namely $a^{-1}\lang \supset ab^{-1} \lang \supset abc^{-1}\lang \supset \ldots$
Figure~\ref{fig:nondet} shows this descending chain, where we have omitted languages like $aa^{-1}\currlang$, as they only differ from $a^{-1}\currlang$ on the empty word.
The existence of a such a chain implies that $\currlang$ cannot be accepted by a residual automaton. This is a consequence of Theorem~\ref{thm:residual-characterisation}, as we shall see later.

\begin{figure}[h!]
\centering
\begin{tikzpicture}[automaton, baseline=0.1cm]
\node[state] (q0) at (0,0) {$a$};
\node[state, accepting] (q1) at (1,0) {};
\node (init) at (-1.5, 0) {$\aut_\up{n} = $};
\node (init2) at (1.66, 0) {};

\path[tr]
(init) edge node {guess $a$} (q0)
(init2) edge (q1)
(q0) edge node {$a$} (q1)
(q0) edge[loop above] node {$\neq a$} (q0);
\end{tikzpicture}
\hspace{2cm}
\begin{tikzpicture}[poset, baseline=(current bounding box.center)]
\node (L) at (0,0) {$\currlang$};
\node (aL) at (-0.7, -1) {$a^{-1}\currlang$};
\node (cL) at (0, -1) {$\cdots$};
\node (bL) at (0.7, -1) {$b^{-1}\currlang$};
\node (abL) at (0, -2) {$ab^{-1}\currlang$};
\node (cdL) at (-0.95, -2) {$\cdots$};
\node (dcL) at (0.95, -2) {$\cdots$};
\node (abcL) at (0, -3) {$abc^{-1}\currlang$};
\node (cdeL) at (-1, -3) {$\cdots$};
\node (edcL) at (1, -3) {$\cdots$};
\node (inf) at (0, -4.0) {};
\node (infl) at (-1, -4.0) {};
\node (infr) at (1, -4.0) {};
\draw (L) -- (aL) (L) -- (bL) (aL) -- (abL) (bL) -- (abL) (abL) -- (abcL);
\draw[dashed] (abcL) -- (inf);
\begin{scope}[on background layer]
\draw[black!15] (L) -- (cL) (cL) -- (cdL) (cL) -- (dcL) (aL) -- (cdL) (bL) -- (dcL) (cdL) -- (cdeL) (cdL) -- (abcL) (abL) -- (cdeL) (abL) -- (edcL) (dcL) -- (abcL) (dcL) -- (edcL);
\draw[dashed, black!15] (cdeL) -- (infl) (edcL) -- (infr);
\node [orbit, fit=(L)] {};
\node [orbit, fit=(aL)(bL)(cL)] {};
\node [orbit, fit=(abL)(cdL)(dcL)] {};
\node [orbit, fit=(abcL)(cdeL)(edcL)] {};
\end{scope}
\end{tikzpicture}
\caption{A nondeterministic automaton accepting $\currlang$, and the poset $\Res(\currlang)$.}
\label{fig:nondet}
\end{figure}

\renewcommand{\currlang}{\lang_\up{r}}
\subsection*{Residual: Last letter is unique but anchored}
We reconsider the previous automaton $\aut_\up{n}$ and add a transition in order to make the automaton residual.
First, we extend the alphabet to $\Sigma = \atoms \cup \left\{ \Anchor(a) \mid a \in \atoms \right\}$, where $\Anchor$ is nothing more than a label.
Then, we add the transitions $(a, \Anchor(a), a)$ for each $a \in \atoms$, see Figure~\ref{fig:res}.
The language accepted by the new automaton is
\[ \currlang = \{ wa \mid a \in \atoms, w \in (\atoms \setminus \{a\} \cup \{ \Anchor(a) \})^* \} \ch{\cup \{ \epsilon \}} \]
Here, we have forced the automaton to be residual, by adding an anchor to the first state.
Nevertheless, guessing is still necessary.
In the poset, we note that all elements in the descending chain can now be obtained as unions of $\Anchor(a)^{-1}\currlang$.
For instance, $a^{-1}\currlang = \bigcup_{b \neq a} \Anchor(b)^{-1}\currlang$.
Note that $\Anchor(a)\Anchor(b)^{-1}\currlang = \emptyset$ and $\Anchor(a)a^{-1}\currlang = \{ \epsilon \}$.

\begin{figure}[h!]
\centering
\begin{tikzpicture}[automaton, baseline=-0.1cm]
\node[state] (q0) at (0,0) {$a$};
\node[state, accepting] (q1) at (1,0) {};
\node (init) at (-1.5, 0) {$\aut_\up{r} = $};
\node (init2) at (1.66, 0) {};

\path[tr]
(init) edge node {guess $a$} (q0)
(init2) edge (q1)
(q0) edge node {$a$} (q1)
(q0) edge[loop above] node {$\neq a$} (q0)
(q0) edge[loop below] node {$\Anchor(a)$} (q0);
\end{tikzpicture}
\hspace{2cm}
\begin{tikzpicture}[poset, baseline=(current bounding box.center)]
\node (L) at (0,0) {$\currlang$};
\node (aL) at (-0.7, -1) {$a^{-1}\currlang$};
\node (cL) at (0, -1) {$\cdots$};
\node (bL) at (0.7, -1) {$b^{-1}\currlang$};
\node (abL) at (0, -2) {$ab^{-1}\currlang$};
\node (cdL) at (-0.95, -2) {$\cdots$};
\node (dcL) at (0.95, -2) {$\cdots$};
\node (abcL) at (0, -3) {$abc^{-1}\currlang$};
\node (cdeL) at (-1, -3) {$\cdots$};
\node (edcL) at (1, -3) {$\cdots$};
\node (anc) at (-0.91, -5) {$\Anchor(c)^{-1}\currlang$};
\node (anca) at (1, -5) {$\Anchor(a)a^{-1}\currlang$};
\node (ancanc) at (0, -6) {$\Anchor(c)\Anchor(d)^{-1}\currlang$};
\draw (L) -- (aL) (L) -- (bL) (aL) -- (abL) (bL) -- (abL) (abL) -- (abcL);
\draw (anc) -- (aL) (anc) -- (bL) (anc) -- (abL) (anca) -- (bL) (ancanc) -- (anc) (ancanc) -- (anca);
\draw[dashed] (abcL) -- (ancanc);
\begin{scope}[on background layer]
\draw[black!15] (L) -- (cL) (cL) -- (cdL) (cL) -- (dcL) (aL) -- (cdL) (bL) -- (dcL) (cdL) -- (cdeL) (cdL) -- (abcL) (abL) -- (cdeL) (abL) -- (edcL) (dcL) -- (abcL) (dcL) -- (edcL);
\draw[black!15] (anc) -- (cdeL) (anc) -- (edcL) (anc) -- (cdL) (anc) -- (dcL) (anc) -- (cL) (anca) -- (cdeL) (anca) -- (cdL) (anca) -- (dcL) (anca) -- (edcL) (anca) -- (cL);
\draw[dashed, black!15] (cdeL) -- (ancanc) (edcL) -- (ancanc);
\node [orbit, fit=(L)] {};
\node [orbit, fit=(aL)(bL)(cL)] {};
\node [orbit, fit=(abL)(cdL)(dcL)] {};
\node [orbit, fit=(abcL)(cdeL)(edcL)] {};
\node [orbit, draw=red, fit=(anc)] {};
\node [orbit, draw=red, fit=(anca)] {};
\node [orbit, fit=(ancanc)] {};
\end{scope}
\end{tikzpicture}
\caption{A residual automaton accepting $\currlang$, and the poset $\Res(\currlang)$.}
\label{fig:res}
\end{figure}

\renewcommand{\currlang}{\lang_\up{ng}}
\subsection*{Non-guessing nondeterministic: Repeated atom with different successor}
Consider the language
\[ \currlang \coloneqq \left\{ u a b v a c \mid u,v \in \atoms^{*}, a,b,c \in \atoms, \ch{b \neq c} \right\} . \]
Here, we allow $a=b$ or $a=c$ in this definition.
This is a language which can be accepted by a non-guessing automaton (Figure~\ref{fig:ng-nondet}).
However, there is no residual automaton for this language.
The poset structure of $\Res(\currlang)$ is very complicated.
We will return to this example after Theorem~\ref{thm:residual-characterisation}.

\begin{figure}[h!]
\centering
\begin{tikzpicture}[automaton, baseline=-.6cm]
\node[state, initial, initial text={$\aut_\up{ng} =$}] (q0) at (0,0) {};
\node[state] (q1) at (1,0) {$a$};
\node[state] (q2) at (0,-1) {$ab$};
\node[state] (q3) at (1,-1) {$b$};
\node[state,accepting] (q4) at (2,-1) {};

\path[tr]
(q0) edge[loop above] node {$\atoms$} (q0)
(q0) edge node {$a$} (q1)
(q2) edge[loop above] node {$\atoms$} (q2)
(q2) edge node {$a$} (q3)
(q3) edge node {$\neq b$} (q4);
\draw[tr, rounded corners=15pt] (q1) -- node {$b$} (2,0) -- (q2);
\end{tikzpicture}\hspace{2cm}
\begin{tikzpicture}[poset, x=2.55cm, baseline=(current bounding box.center)]
\node (L) at (0,0) {$\currlang$};
\node (aL) at (-0.33, 1) {$a^{-1}\currlang$};
\node (bL) at (0.33, 1) {$b^{-1}\currlang$};
\node (aaL) at (-1, 2) {$aa^{-1}\currlang$};
\node (baL) at (0, 2) {$ba^{-1}\currlang$};
\node (abL) at (1, 2) {$ab^{-1}\currlang$};
\node (abaL) at (-0.33, 3) {$aba^{-1}\currlang$};
\node (cbaL) at (0.33, 3) {$cba^{-1}\currlang$};
\node (inf1) at (-1, 3.73) {};
\node (inf2) at (-0.33, 3.73) {};
\node (inf3) at (0.33, 3.73) {};
\node (inf4) at (1, 3.73) {};
\draw (L) -- (aL) (L) -- (bL) (aL) -- (aaL) (aL) -- (baL) (bL) -- (abL) (baL) -- (abaL) (baL) -- (cbaL);
\draw[dashed] (inf1) -- (aaL) (inf2) -- (abaL) (inf3) -- (cbaL) (inf4) -- (abL);
\begin{scope}[on background layer]
\node [orbit, fit=(L)] {};
\node [orbit, fit=(aL)(bL)] {};
\node [orbit, fit=(aaL)] {};
\node [orbit, fit=(baL)] {};
\node [orbit, fit=(abL)] {};
\node [orbit, fit=(abaL)] {};
\node [orbit, fit=(cbaL)] {};
\end{scope}
\end{tikzpicture}
\caption{A deterministic automaton accepting $\currlang$, and the poset $\Res(\currlang)$.}
\label{fig:ng-nondet}
\end{figure}

\section{Canonical Residual Nominal Automata}
\label{sec:crna}

In this section we will give a characterisation of \emph{canonical} residual automata. We will first introduce notions of nominal lattice theory, then we will state our main result~(Theorem~\ref{thm:residual-characterisation}).
We conclude the section by providing similar results for non-guessing automata.

\subsection{Theory of Nominal Join-Semilattices}
\label{sec:nominal-join-semilattices}

We abstract away from words and languages and consider the set $\Powfs(Z)$ for an arbitrary nominal set $Z$.
This is a Boolean algebra of which the operations $\wedge, \vee, \neg$ are all equivariant maps~\cite{GabbayLP11}.
Moreover, the finitely supported union
\[
\bigvee \colon \Powfs(\Powfs(Z)) \to \Powfs(Z)
\]
is also equivariant.
We note that this is more general than a binary union, but it is not a complete join-semilattice.
Hereafter, we shall denote set inclusion by $\leq$ ($<$ when strict).

\begin{definition}
Given a nominal set $Z$ and $X \subseteq \Powfs(Z)$ equivariant\footnote{A similar definition could be given for finitely supported $X$. In fact, all results in this section generalise to finitely supported. But we use equivariance for convenience.},
we define the set generated by $X$ as
\[
\clos{X} \coloneqq \left\{ \bigvee \mathfrak{x} \mid \mathfrak{x} \subseteq X \text{ finitely supported} \right\} \subseteq \Powfs(Z).
\]
\end{definition}

\begin{remark}
The set $\clos{X}$ is closed under the operation $\bigvee$, and moreover is the smallest equivariant set closed under $\bigvee$ containing $X$. In other words, $\clos{-}$ defines a closure operator.
We will often say ``$X$ generates $Y$'', by which we mean $Y \subseteq \clos{X}$.
\end{remark}

\begin{definition}
\label{def:join-irred}
Let $X \subseteq \Powfs(Z)$ equivariant and $x \in X$, we say that $x$ is \emph{join-irreducible in $X$} if it is non-empty and
\[ x = \bigvee \mathfrak{x} \quad \implies \quad x \in \mathfrak{x} , \]
for every finitely supported $\mathfrak{x} \subseteq X$.
The subset of all join-irreducible elements is denoted by
\[
\JI(X) \coloneqq \left\{ x \in X \mid x \text{ join-irreducible in } X \right\}.
\]
This is again an equivariant set.
For convenience, we may use the following equivalent definition of join-irreducible:
$x$ is non-empty and $( \forall x_0 \in \mathfrak{x} . x_0 < x )\implies \bigvee \mathfrak{x} < x $, \ch{for every finitely supported $\mathfrak{x} \subseteq X$.}
\end{definition}

\begin{remark}
In lattice and order theory, join-irreducible elements are usually defined only for a lattice (see, e.g.,~\cite{DaveyP02}).
However, we define them for arbitrary subsets of a lattice.
(\ch{Note that a subset of a lattice is not a sub-lattice in general.})
This generalisation will be needed later, when we consider the poset $\Res(\lang)$, \ch{which is contained in the lattice $\Powfs(\Sigma^*)$, but it is not a sub-lattice.}
\end{remark}

\begin{remark}
The notion of join-irreducible, as we have defined here, corresponds to the notion of \emph{prime} in some papers on learning nondeterministic automata~\cite{BolligHKL09, DenisLT02, MoermanS0KS17}.
Unfortunately, the word \emph{prime} has a slightly different meaning in lattice theory.
We stick to the terminology of lattice theory.
\end{remark}

If a set $Y$ is well-behaved, then its join-irreducible elements will actually generate the set $Y$.
This is normally proven with a descending chain condition.
We first restrict our attention to orbit-finite sets.
The following Lemma extends~\cite[Lemma~2.45]{DaveyP02} to the nominal setting.
\ch{The proof is analogous to the ordinary case, except that is relies on the specific structure of equality atoms (Lemma~\ref{lem:minimal-in-poset}).}

\begin{lemma}\label{lem:join-of-ji}
Let $X \subseteq \Powfs(Z)$ be an orbit-finite and equivariant set.
\begin{enumerate}
\item
Let $a \in X, b \in \Powfs(Z)$ and $a \not\leq b$.
Then there is $x \in \JI(X)$ such that $x \leq a$ and $x \not\leq b$.
\item
Let $a \in X$, then $a = \bigvee \left\{ x \in X \mid x \text{ join-irreducible in } X \text{ and } x \leq a \right\}$.
\end{enumerate}
\end{lemma}
\begin{proof}
\textit{Ad 1.}
Consider the set $S = \left\{ x \in X \mid x \leq a, x \not\leq b \right\}$.
This is a finitely supported and $\supp(S)$-orbit-finite set, hence it has some minimal element $m \in S$ by Lemma~\ref{lem:minimal-in-poset} (here we are using its generalisation to \ch{finitely-supported sets}).
We shall prove that $m$ is join-irreducible in $X$.
Let $\mathfrak{x} \subseteq X$ finitely supported and assume that $x_0 < m$ for each $x_0 \in \mathfrak{x}$.
Note that $x_0 < m \leq a$ and so that $x_0 \notin S$ (otherwise $m$ was not minimal).
Hence $x_0 \leq b$ (by definition of $S$).
So $\bigvee \mathfrak{x} \leq b$ and so $\bigvee \mathfrak{x} \notin S$, which concludes that $\bigvee \mathfrak{x} \neq m$, and so $\bigvee \mathfrak{x} < m$ as required.

\textit{Ad 2.}
Consider the set $T = \left\{ x \in \JI(X) \mid x \leq a \right\}$.
This set is finitely supported, so we may define the element $b = \bigvee T \in \Powfs(Z)$.
It is clear that $b \leq a$, we shall prove equality by contradiction.
Suppose $a \not\leq b$, then by \textit{(1.)}, there is a join-irreducible $x$ such that $x \leq a$ and $x \not\leq b$.
By the first property of $x$ we have $x \in T$, so that $x \not\leq b = \bigvee T$ is a contradiction.
We conclude that $a = b$, i.e., $a = \bigvee T$ as required.
\end{proof}

\begin{corollary}
Let $X \subseteq \Powfs(Z)$ be an orbit-finite equivariant subset.
The join-irreducibles of $X$ generate $X$, i.e., $X \subseteq \clos{\JI(X)}$.
\end{corollary}

So far, we have defined join-irreducible elements relative to some fixed set.
We will now show that these elements remain join-irreducible when considering them in a bigger set, as long as the bigger set is generated by the smaller one.
This will later allow us to talk about \emph{the} join-irreducible elements.

\begin{lemma}\label{lem:absolute-ji}
Let $Y \subseteq X \subseteq \Powfs(Z)$ equivariant and suppose that $X \subseteq \clos{\JI(Y)}$.
Then $\JI(Y) = \JI(X)$.
\end{lemma}
\begin{proof}
($\supseteq$)
Let $x \in X$ be join-irreducible in $X$.
Suppose that $x = \bigvee \mathfrak{y}$ for some finitely supported $\mathfrak{y} \subseteq Y$.
Note that also $\mathfrak{y} \subseteq X$
Then $x = y_0$ for some $y_0 \in \mathfrak{y}$, and so $x$ is join-irreducible in $Y$.

($\subseteq$)
Let $y \in Y$ be join-irreducible in $Y$.
Suppose that $y = \bigvee \mathfrak{x}$ for some finitely supported $\mathfrak{x} \subseteq X$.
Note that every element $x \in \mathfrak{x}$ is a union of elements in $\JI(Y)$ (by the assumption $X \subseteq \clos{\JI(Y)}$).
Take $\mathfrak{y}_x = \left\{ y \in \JI(Y) \mid y \leq x \right\}$, then we have $x = \bigvee \mathfrak{y}_x$ and
\[ y = \bigvee \mathfrak{x}
= \bigvee \left\{ \bigvee \mathfrak{y}_x \mid x \in \mathfrak{x} \right\}
= \bigvee \left\{ y_0 \mid y_0 \in \mathfrak{y}_x, x \in \mathfrak{
	x} \right\} . \]
The last set is a finitely supported subset of $Y$, and so there is a $y_0$ in it such that $y = y_0$.
Moreover, this $y_0$ is below some $x_0 \in \mathfrak{x}$, which gives $y_0 \leq x_0 \leq y$.
We conclude that $y = x_0$ for some $x_0 \in \mathfrak{x}$.
\end{proof}

In other words, the join-irreducibles of $X$ are the smallest set generating $X$.

\begin{corollary}\label{cor:smallest}
If an orbit-finite set $Y$ generates $X$, then $\JI(X) \subseteq Y$.
\end{corollary}

\subsection{Characterising Residual Languages}

We are now ready to state and prove the main theorem of this paper.
We fix the alphabet $\Sigma$.
Recall that the nominal Myhill-Nerode theorem tells us that a language is accepted by a deterministic automaton if and only if $\Res(\lang)$ is orbit-finite.
Here, we give a similar characterisation for languages accepted by residual automata.
Moreover, the following result gives a canonical construction.

\begin{theorem}
\label{thm:residual-characterisation}
Given a language $\lang \subseteq \Powfs(\Sigma^*)$, the following are equivalent:
\begin{enumerate}
\item
$\lang$ is accepted by a residual automaton.
\item
There is some orbit-finite set $J \subseteq \Res(\lang)$ which generates $\Res(\lang)$.
\item
The set $\JI(\Res(\lang))$ is orbit-finite and generates $\Res(\lang)$.
\end{enumerate}
\end{theorem}
\begin{proof}
We prove three implications:

\paragraph{($1 \Rightarrow 2$)}
Let $\aut \coloneqq (\Sigma,Q,I,F,\delta)$ be a residual automaton accepting $\lang$.
Take the set of languages accepted by the states: $J \coloneqq \{ \stlang{q} \mid q \in \aut \}$.
This is clearly orbit-finite, since $Q$ is.
Moreover, each derivative is generated as follows:
\[ w^{-1}\lang = \bigvee \left\{ \stlang{q} \mid q \in \delta(I, w) \right\} . \]

\paragraph{($2 \Rightarrow 3$)}
We can apply Lemma~\ref{lem:absolute-ji} with $Y = J$ and $X = \Res(\lang)$ \ch{and obtain $\JI(J) = \JI(\Res(\lang))$}.
It follows that $\JI(\Res(\lang))$ is orbit-finite (since it is a subset of $J$) and generates $\Res(\lang)$.

\paragraph{($3 \Rightarrow 1$)}
\ch{Consider the following residual automaton:}
\begin{align*}
Q &\coloneqq \JI(\Res(\lang)) \\
I &\coloneqq \left\{ w^{-1}\lang \in Q \mid w^{-1}\lang \leq \lang \right\} \\
F &\coloneqq \left\{ w^{-1}\lang \in Q \mid \epsilon \in w^{-1}\lang \right\} \\
\delta(w^{-1}\lang, a) &\coloneqq \left\{ v^{-1}\lang \in Q \mid v^{-1}\lang \leq wa^{-1}\lang \right\}
\end{align*}
Note that $\aut \coloneqq (\Sigma,Q,I,F,\delta)$ is a well-defined nominal automaton.
In fact, all the components are orbit-finite, and equivariance of $\leq$ implies equivariance of $\delta$.

\ch{We shall now prove that the language of this automaton is exactly $\lang$. As a first step, we prove that $\stlang{q} = w^{-1}\lang$ by induction over words.}
\begin{align*}
\epsilon \in \stlang{w^{-1}\lang} &\iff w^{-1}\lang \in F \iff \epsilon \in w^{-1}\lang \\
au \in \stlang{w^{-1}\lang}
&\iff u \in \stlang{\delta(w^{-1}\lang, a)} \\
&\iff u \in \stlang{\left\{ v^{-1}\lang \in Q \mid v^{-1}\lang \leq wa^{-1}\lang \right\}} \\
{}^\text{(i)}
&\iff u \in \bigvee \left\{ v^{-1} \lang \in Q \mid v^{-1}\lang \leq wa^{-1}\lang \right\} \\
&\iff \exists v^{-1}\lang \in Q \text{ with } v^{-1}\lang \leq wa^{-1}\lang \text{ and } u \in v^{-1}\lang \\
{}^\text{(ii)}
&\iff u \in wa^{-1}\lang \iff au \in w^{-1}\lang
\end{align*}
At step (i) we have used the induction hypothesis ($u$ is a shorter word than $au$) and the fact that $\stlang{-}$ preserves unions.
At step (ii, right-to-left) we have used that $v^{-1}\lang$ is join-irreducible.
The other steps are unfolding definitions.

Now, note that $\lang = \bigvee \left\{ w^{-1}\lang \mid w^{-1}\lang \leq \lang, w \in
\Sigma^\star \right\}$, since the join-irreducible languages generate all languages.
In other words, the initial states (together) accept $\lang$.
\end{proof}

\begin{corollary}
The construction above defines a \emph{canonical} residual automaton with the following uniqueness property: it has the minimal number of orbits of states and the maximal number of orbits of transitions.
\end{corollary}
\begin{proof}
State minimality follows from Corollary~\ref{cor:smallest},
where we note that the states of any residual automata accepting $\lang$ form a generating subset of $\Res(\lang)$.
Maximality of transitions follows from the fact that no transitions can be added without changing the language.
\end{proof}

For finite alphabets, the classes of languages accepted by DFAs and NFAs are the same (by determinising an NFA).
This means that $\Res(\lang)$ is always finite if $\lang$ is accepted by an NFA, and we can always construct the canonical RFSA\@.
Here, this is not the case, that is why we need to stipulate (in~Theorem~\ref{thm:residual-characterisation}) that the set $\JI(\Res(\lang))$ is orbit-finite \emph{and} actually generates $\Res(\lang)$.
Either condition may fail, as we will see in Example~\ref{ex:non-residual}.

\begin{example}
In this example we show that residual automata can also be used to compress deterministic automata.
The language $\lang \coloneqq \{ a b b \ldots b \mid a \neq b \}$ can be accepted by a deterministic automaton of 4 orbits, and this is minimal.
(A zero amount of $b$s is also accepted in $\lang$.)
The minimal residual automaton, however, has only 2 orbits, given by the join-irreducible languages:
\begin{align*}
\epsilon^{-1} \lang &= \{ a b b \ldots b \mid a \neq b \} \\
ab^{-1} \lang &= \{ b b \ldots b \} \qquad (a,b \in \atoms \text{ distinct})
\end{align*}
The trick in defining the automaton is that the $a$-transition from $\epsilon^{-1} \lang$ to $ab^{-1} \lang$ \emph{guesses} the value $b$.
In the next section (Section~\ref{sec:non-guessing}), we will define the canonical \emph{non-guessing} residual automaton, which has 3 orbits.
\end{example}

\begin{example}\label{ex:non-residual}
We return to the examples $\lang_\up{n}$ and $\lang_\up{ng}$ from Section~\ref{sec:examples}.
We claim that neither language can be accepted by a residual automaton.

\renewcommand{\currlang}{\lang_\up{n}}
For $\lang_\up{n}$ we note that there is an infinite descending chain of derivatives
\[ \currlang > a^{-1}\currlang > ab^{-1}\currlang > abc^{-1}\currlang > \cdots \]
Each of these languages can be written as a union of smaller derivatives.
For instance, $a^{-1}\currlang = \bigcup_{b \neq a} ab^{-1}\currlang$.
This means that $\JI(\Res(\currlang)) = \emptyset$, hence it does not generate $\Res(\currlang)$ and by Theorem~\ref{thm:residual-characterisation} there is no residual automaton.

\renewcommand{\currlang}{\lang_\up{ng}}
In the case of $\currlang$, we have an infinite ascending chain
\begin{equation}
\currlang < a^{-1}\currlang < ba^{-1}\currlang < cba^{-1}\currlang < \cdots
\label{eq:chain}
\end{equation}
This in itself is not a problem: the language $\lang_\up{ng,r}$ also has an infinite ascending chain.
However, for $\currlang$, none of the languages in this chain are a union of smaller derivatives, which we shall now prove formally.
\begin{claim}
All the languages in \eqref{eq:chain} are join-irreducible.
\end{claim}

\begin{proof}
Consider the word $w = a_k \ldots a_1 a_0$ with $k \geq 1$ and all $a_i$ distinct atoms.
We will prove that $w^{-1} \currlang$ is join-irreducible in $\Res(\currlang)$, by considering all $u^{-1} \currlang \subseteq w^{-1} \currlang$.

Observe that if $u$ is a suffix of $w$, then $u^{-1} \currlang \subseteq w^{-1} \currlang$.
This is easily seen from the given automaton, since it may skip any prefix.
We now show that $u$ being a suffix of $w$ is also a necessary condition.

\ch{Assume that $u$ is not a suffix of $w$, so there is an $i \geq 0$ with $x \neq a_i$ and $u$ contains the suffix $x a_{i-1} \ldots a_0$.
Take a fresh atom $a_{-1}$.
If $x = a_k$ for some $k$, let $c \coloneqq a_{k-1}$ (note that we may use $a_{-1}$ here) and otherwise let $c$ be fresh.
Then $a_{-1} x c$ is in $u^{-1}\lang$, since we have repeated $x$ with a different successor.
However, regarding $w^{-1}\lang$: If $x$ does not occur in $w$, then $c$ is fresh and $a_{-1} x c$ is clearly not in $w^{-1}\lang$ (all atoms are distinct).
If $x = a_k$ (and so $c = a_{k-1}$), then $w a_{-1} a_k a_{k-1}$ mentions only $a_k$ and $a_{k-1}$ twice, but not with distinct successors; hence $a_{-1} x c \notin w^{-1}\lang$.
We conclude that if $u$ is not a suffix of $w$, then $u^{-1}\lang$ is not a subset of $w^{-1}\lang$.}

So far, we have shown that
\[ \{ u \mid u^{-1} \currlang \subseteq w^{-1} \currlang \} = \{ u \mid u~\text{is a suffix of}~w\}. \]
To see that $w^{-1} \currlang$ is indeed join-irreducible, we consider the join $X = \bigvee \{ u^{-1} \currlang \mid u~\text{is a strict suffix of}~w \}$.
Note that $a_k a_k \notin X$, but $a_k a_k \in w^{-1} \currlang$.
We conclude that $w^{-1} \currlang \neq \bigvee \{ u^{-1} \currlang \mid u^{-1} \currlang \subsetneqq w^{-1} \currlang \}$ as required.
\end{proof}
This result implies that the set $\JI(\Res(\currlang))$ is \emph{not orbit-finite}.
By Theorem~\ref{thm:residual-characterisation}, we can conclude that there is no residual automaton accepting $\currlang$.

\end{example}

\begin{remark}
For arbitrary (nondeterministic) languages there is also a characterisation in the style of Theorem~\ref{thm:residual-characterisation}.
Namely, $\lang$ is accepted by an automaton iff there is an orbit-finite set $Y \subseteq \Powfs(\Sigma^*)$ which generates the derivatives.
However, note that the set $Y$ need not be a subset of the set of derivatives.
In these cases, we do not have a canonical construction for the automaton.
Different choices for $Y$ define different automata and there is no way to pick $Y$ naturally.
\end{remark}

\subsection{Automata without guessing}
\label{sec:non-guessing}

We reconsider the above results for non-guessing automata.
Nondeterminism in nominal automata allows naturally for guessing, meaning that the automaton may store symbols in registers without explicitly reading them.
\ch{For instance, in Figure~\ref{fig:nondet} the automaton non-deterministically stores a(ny) symbol in the initial state without actually reading it, and by doing so it ``guesses'' which symbol will be read at the end of a word.}
The original definition of register automata in~\cite{KaminskiF94} does not allow for guessing, and non-guessing automata remain actively researched~\cite{MottetQ19}.
Register automata with guessing were introduced in~\cite{KaminskiZ10}, because it was realised that non-guessing automata are not closed under reversal.

\ch{
To adapt our theory to non-guessing automata, we need to introduce a more restricted form of powerset.
We say that $U \subseteq X$ is \emph{uniformly finitely supported} (ufs in short) if $\cup_{x \in U} \supp(x)$ is finite.
The \emph{ufs powerset} is defined as follows:
\[
	\Powufs(X) = \{ U \subseteq X \mid U \; \text{is uniformly finitely supported} \}
\]
This too comes with its notion of ufs-join, performing the union of ufs sets.

The key insight for this section is that the constraints on supports for non-guessing automata (see Definition~\ref{def:nom-aut}) imply that the transition relation can be expressed as a function of the form $\delta \colon Q \times \Sigma \to \Powufs(Q)$.
Intuitively, whenever a symbol $a$ is read from a state $q$, all successor states must have support that is at most that of $q$ plus that of $a$, which implies that the union of their supports is finite (i.e., they form a ufs set).
The consequence of shifting from $\Powfs$ to $\Powufs$ is that, when giving a specialised version of Theorem~\ref{thm:residual-characterisation} for non-guessing automata, we can consider the join-semilattice structure given by ufs sets and ufs unions. We first characterise join-irreducibles for such join-semilattices.}

\begin{definition}
Let $X \subseteq \Powfs(Z)$ be equivariant and $x \in X$, we say that $x$ is \ch{\emph{ufs-join-irreducible in $X$}} if
$x = \bigvee \mathfrak{x} \implies x \in \mathfrak{x}$,
for every finitely supported $\mathfrak{x} \subseteq X$ such that $\supp(x_0) \subseteq \supp(x)$, for each $x_0 \in \mathfrak{x}$.
The set of all \ch{ufs-join-irreducible} elements is denoted by
\[
\ch{\JIufs(X) \coloneqq \left\{ x \in X \mid x \text{ ufs-join-irreducible in } X \right\}.}
\]
\end{definition}
The only change required is an additional condition on the elements and supports in $\mathfrak{x}$.
In particular, the sets $\mathfrak{x}$ are ufs sets, hence their union is ufs.

All the lemmas from the previous section are proven similarly.
We state the main result for non-guessing automata.

\begin{theorem}
\label{thm:non-guessing-residuals}
Given a language $\lang \subseteq \Powfs(\Sigma^*)$, the following are equivalent:
\begin{enumerate}
\item
$\lang$ is accepted by a non-guessing residual automaton.
\item
There is some orbit-finite set $J \subseteq \Res(\lang)$ which generates $\Res(\lang)$ by ufs unions.
\item
The set $\JIufs(\Res(\lang))$ is orbit-finite and generates $\Res(\lang)$ by ufs unions.
\end{enumerate}
\end{theorem}
\begin{proof}[Proof Sketch]
The proof is similar to that of Theorem~\ref{thm:residual-characterisation}, we briefly sketch each direction.

\ch{For direction ($1 \Rightarrow 2$), we observe that, for a residual non-guessing automaton $\aut \coloneqq (\Sigma,Q,I,F,\delta)$, we have $\supp(q) \subseteq \supp(w)$, for $q \in \delta(I,w)$. Since $\supp(\stlang{q}) \subseteq \supp(q)$, the set $\left\{ \stlang{q} \mid q \in \delta(I, w) \right\}$ is ufs, and its ufs union gives $w^{-1}\lang$.

For direction ($2 \Rightarrow 3$), it is easy to see that Lemma~\ref{lem:absolute-ji} applies to generation via ufs unions, taking $\JIufs(-)$ as join-irreducibles. Therefore we have $\JIufs(J) = \JIufs(\Res(\lang))$.}

For direction ($3 \Rightarrow 1)$ we need a slightly different definition of the canonical automaton:
\begin{align*}
Q &\coloneqq \JIufs(\Res(\lang)) \\
I &\coloneqq \left\{ w^{-1}\lang \in Q \mid w^{-1}\lang \leq \lang, \supp(w^{-1}\lang) \subseteq \supp(\lang) \right\} \\
F &\coloneqq \left\{ w^{-1}\lang \in Q \mid \epsilon \in w^{-1}\lang \right\} \\
\delta(w^{-1}\lang, a) &\coloneqq \left\{ v^{-1}\lang \in Q \mid v^{-1}\lang \leq wa^{-1}\lang , \supp(v^{-1}\lang) \subseteq \supp(wa^{-1}\lang) \right\}
\end{align*}
\ch{
The fact that this automaton accepts $\lang$ can be proven similarly to what done for Theorem~\ref{thm:residual-characterisation}.
We need to show that this automaton is indeed non-guessing, namely:
\begin{enumerate}
	\item $\supp(I) = \emptyset$;
	\item $v^{-1}\lang \subseteq \supp(a) \cup \supp(w^{-1}\lang)$, for each $v^{-1}\lang \in \delta(w^{-1}\lang, a)$.
\end{enumerate}
The first condition follows from $\lang$ being equivariant. For the second one, we have
\begin{align*}
	\supp(v^{-1}\lang) \; \stackrel{\text{(i)}}{\subseteq} \; \supp(wa^{-1}\lang) \; \stackrel{\text{(ii)}}{\subseteq} \; \supp(a) \cup \supp(w^{-1}\lang)
\end{align*}
where (i) is by the definition of $\delta(w^{-1}\lang, a)$ and (ii) follows by equivarience of the function $(w^{-1}\lang,a) \mapsto  wa^{-1}\lang \colon \Res(\lang) \times \Sigma  \to \Res(\lang)$ (see, e.g., \cite[Lemma 2.12]{Pitts13}).
}
\end{proof}

To better understand the structure of the canonical non-guessing residual automaton, we recall the following fact.
\begin{lemma}\label{lem:finite-powerset}
For orbit-finite nominal sets $Q$, we have $\Powufs(Q) = \Powfin(Q)$.
\end{lemma}

As a consequence, the transition function of non-guessing automata can be written as $\delta \colon Q \times \Sigma \to \Powfin(Q)$. This shows that the canonical non-guessing residual automaton has finite nondeterminism.
It also shows that it is sufficient to consider \emph{finite unions} in Theorem~\ref{thm:non-guessing-residuals}, instead of uniformly supported ones.

\section{Decidability and Closure Results}
\label{sec:dec}

In this section we investigate decidability and closure properties of residual automata.
First, a positive result: universality is decidable for residual automata.
This is in contrast to the nondeterministic case, where universality is undecidable, even for non-guessing automata~\cite{Bojanczyk19}.

In the constructions below, we use \emph{computation with atoms.}
This is a computation paradigm which allow algorithmic manipulation of infinite --- but orbit-finite --- nominal sets. For instance, it allows looping over such a set in finite time.
Important here is that this paradigm is equivalent to regular computability (see~\cite{BojanczykT18}) and implementations exist to compute with atoms~\cite{KlinS16,KopczynskiT17}.

\begin{proposition}\label{prop:decidability-universality}
Universality for residual nominal automata is decidable.
Formally: given a residual automaton $\aut$, it is decidable whether $\lang(\aut) = \Sigma^*$.
\end{proposition}
\begin{proof}
We will sketch an algorithm that, given a residual automaton $\aut$, answers whether $\lang(\aut) = \Sigma^*$.
The algorithm decides \emph{negatively} in the following cases:
\begin{itemize}
\item
$I = \emptyset$.
In this case the language accepted by $\aut$ is empty.
\item
Suppose there is a $q \in Q$ with $q \notin F$.
By residuality we have $\stlang{q} = w^{-1}\lang(\aut)$ for some $w$.
Note that $q$ is not accepting, so that $\epsilon \notin w^{-1}\lang(\aut)$.
Put differently: $w \notin \lang(\aut)$.
(We note that $w$ is not used by the algorithm. It is only needed for the correctness.)
\item
Suppose there is a $q \in Q$ and $a \in \Sigma$ such that $\delta(q, a) = \emptyset$.
Again $\stlang{q} = w^{-1}\lang(\aut)$ for some $w$.
Note that $a$ is not in $\stlang{q}$.
This means that $wa$ is not in the language.
\end{itemize}
When none of these three cases hold, the algorithm decides \emph{positively}.
We shall prove that this is indeed the correct decision.
If none of the above conditions hold, then $I \neq \emptyset$, $Q = F$, and for all $q \in Q, a \in \Sigma$ we have $\delta(q, a) \neq \emptyset$.
Here we can prove that the language of each state is $\stlang{q} = \Sigma^*$.
Given that there is an initial state, the automaton accepts $\Sigma^*$.

Note that the operations on sets performed in the above cases all terminate, because all involve orbit-finite sets.
\end{proof}

Next we consider equivalence of residual automata and checking whether an automaton is residual.
Both will turn out to be undecidable and we use following construction in order to prove this.

\begin{construction}
Let $\aut = (\Sigma, Q, I, F, \delta)$ be a nondeterministic automaton.
Let
\[ \Sigma' \coloneqq \Sigma \cup \{ \underline{q} \mid q \in Q \} \cup \{ q \mid q \in Q \} \]
be an extended alphabet, where we assume the new symbols $\underline{q}$ and $q$ to be disjoint from $\Sigma$.
We now construct two residual automata from $\aut$, where those symbols are used as anchors:
\begin{align*}
\anc{\aut} &= (\Sigma', \anc{Q}, \anc{I}, \anc{F}, \anc{\delta}) \text{, where}
      & \tp{\aut} &= (\Sigma', \tp{Q}, \tp{I}, \tp{F}, \tp{\delta}) \text{, where} \\
\anc{Q} &= Q \cup \{ \underline{q} \mid q \in Q \}
      & \tp{Q} &= \anc{Q} \cup \{ \top \} \\
\anc{I} &= I \cup \{ \underline{q} \mid q \in Q \}
      & \tp{I} &= \{ \top \} \cup \{ \underline{q} \mid q \in Q \} \\
\anc{F} &= F
      & \tp{F} &= F \cup \{ \top \} \\
\anc{\delta} &= \delta \cup \{ (\underline{q}, \underline{q}, \underline{q}) \mid q \in Q \}
      & \tp{\delta} &= \anc{\delta} \cup \{ (\top, a, \top) \mid a \in \Sigma \} \\
 & \quad \cup \{ (\underline{q}, q, q) \mid q \in Q \}
\end{align*}
Note that these constructions are effective, as they involve computations over orbit-finite sets.
We observe the following facts about $\anc{\aut}$ and $\tp{\aut}$:
\begin{enumerate}
\item
The states $\underline{q}$ and $q$ (in both automata) are anchored by the words $\underline{q}$ and $q$ respectively.
Moreover, any symbol from the original alphabet $a \in \Sigma$ is a characterising word for the state $\top$ in $\tp{\aut}$.
So we conclude that both automata are residual.\label{c1}
\item
For $q \in Q$ we note that the languages $\stlang{\anc{q}}$ and $\stlang{\tp{q}}$ are the same (where $\anc{q} \in \anc{Q}$ and $\tp{q} \in \tp{Q}$ denote the ``same'' state).
Similarly we have $\stlang{\anc{\underline{q}}} = \stlang{\tp{\underline{q}}}$.\label{c2}
\item
If we restrict the alphabet to the original alphabet, we get
\[
\lang(\anc{\aut}) \cap \Sigma^* = \lang(\aut) \qquad \text{and} \qquad
\lang(\tp{\aut}) \cap \Sigma^* = \Sigma^* .
\]
\label{c3}
\end{enumerate}
\label{cons:res}
\end{construction}

\begin{proposition}\label{prop:undecidability-equivalence}
Equivalence of residual automata is undecidable.
\end{proposition}
\begin{proof}
We show undecidability by reducing the universality problem for nondeterministic automata to the equivalence problem.
(Note that a reduction from universality of residual automata will not work as that is decidable.)
We use the above construction and prove that $\aut$ is universal if and only if $\anc{\aut}$ and $\tp{\aut}$ are equivalent.

Suppose $\lang(\aut) = \Sigma^*$, then we get $\lang(\anc{\aut}) \cap \Sigma^* = \lang(\aut) = \Sigma^* = \lang(\tp{\aut}) \cap \Sigma^*$ by (\ref{c3}).
When considering all words on $\Sigma'^*$, we note that the anchors will lead to single states $q$ which accept the same languages by (\ref{c2}).
So $\lang(\anc{\aut}) = \lang(\tp{\aut})$ as required.

Conversely, suppose that $\lang(\anc{\aut}) = \lang(\tp{\aut})$.
Then by (\ref{c3}) we can conclude that
\[ \lang(\aut) = \lang(\anc{\aut}) \cap \Sigma^* = \lang(\tp{\aut}) \cap \Sigma^* = \Sigma^* . \]
We conclude that we can decide universality of nondeterministic automata via equivalence of residual automata.
So equivalence of residual automata is undecidable.
\end{proof}

Last, determining whether an automaton is actually residual is undecidable.
In other words, residuality cannot be characterised as a syntactic property.
This adds value to learning techniques, as they are able to provide automata that are residual by construction.

\begin{proposition}\label{prop:undecidability-residuality}
The problem of determining whether a given nondeterministic nominal automaton is residual is undecidable.
\end{proposition}
\begin{proof}
The construction is inspired by~\cite[Proposition~8.4]{DenisLT02}.\footnote{They prove that checking residuality for NFAs is \textsc{PSpace}-complete via a reduction from universality. Instead of using NFAs, they use a union of $n$ DFAs. This would not work in the nominal setting.}
We show undecidability by reducing the universality problem for nominal automata to the residuality problem.

\begin{figure}\centering
\begin{tikzpicture}[automaton, y=1cm]

\begin{scope}[color=black,local bounding box=A]
\node[state] (op) at (4, -2) {$p$};
\node[state] (oq) at (4.5, -0.5) {$q$};
\node (a1) at (2.8, -0.7) {};
\node (a2) at (2.5, -1.6) {};
\node[yshift=.2cm, draw, ellipse, fit=(a1)(a2)(op)(oq), inner sep=0pt] (aut) {};
\node[yshift=-.33cm] (a) at (aut.north) {$\aut$};
\end{scope}
\begin{scope}[color=black,local bounding box=res]
\node[initial, state, initial where=right, right=1 of op] (up) {$\underline{p}$};
\node[initial, state, initial where=right, right=1 of oq] (uq) {$\underline{q}$};
\path[tr]
(up) edge[near start] node {$p$} (op)
(uq) edge[near start] node {$q$} (oq)
(up) edge[loop above] node[inner sep=4pt] {$\underline{p}$} (up)
(uq) edge[loop above] node[inner sep=4pt] {$\underline{q}$} (uq);
\end{scope}
\begin{scope}[color=red]
\node[initial, state] (x) at (0,0) {$x$};
\node[accepting, state] (t) at (0, -1.5) {$z$};
\node[initial, state] (y) at (1,0) {$y$};

\path[tr]
(x) edge[loop above] node {$\#$} (x)
(x) edge node {$\$$} (t)
(t) edge[loop below] node {$\Sigma$} (t)
(y) edge[loop above] node {$\#$} (y)
(y) edge[bend right, near start] node {$\$$} (a1)
(y) edge[bend right, near start, below] node {$\$$} (a2);
\end{scope}

\node[rectangle, dashed, draw=black, fit=(A) (res),inner sep=2mm, label={$\anc{\aut}$}] (Fit) {};
\end{tikzpicture}

\caption{Sketch of the automaton $\aut'$ constructed in the proof of Proposition~\ref{prop:undecidability-residuality}.}
\label{fig:proof-undecidability-residuality}
\end{figure}

Let $\aut = (\Sigma, Q, I, F, \delta)$ be a nominal (nondeterministic) automaton on the alphabet $\Sigma$.
We apply Construction~\ref{cons:res} and extend the alphabet $\Sigma'$ further by
\[ \Sigma'' \coloneqq \Sigma' \cup \hl{\{ \$, \# \}} , \]
where we assume $\{\$, \#\}$ to be disjoint from $\Sigma$.
We define $\aut' = (\Sigma'', Q', I', F', \delta')$ by
\begin{align*}
Q' &= \anc{Q} \cup \hl{\{ x, y, z \}} \\
I' &= \anc{I} \cup \hl{\{ x, y \}} \\
F' &= \anc{F} \cup \hl{\{ z \}} \\
\delta'
&= \anc{\delta} \cup \hl{\{ (x, \$, z), (x, \#, x), (y, \#, y) \}
\cup \{ (z, a, z) \mid a \in \Sigma \}
\cup \{ (y, \$, q) \mid q \in I \} }
\end{align*}
See Figure~\ref{fig:proof-undecidability-residuality} for a sketch of the automaton $\aut'$.
This is built by adding the \hl{red} part to $\anc{\aut}$.
The key players are states $x$ and $y$ with their languages $\stlang(y) \subseteq \stlang(x)$.
Note that their languages are equal if and only if $\aut$ is universal.

Before we assume anything about $\aut$, let us analyse $\aut'$.
In particular, let us consider whether the residuality property holds for each state.
From (\ref{c1}) we know that this holds for $\anc{\aut}$.
For the states $x$ and $z$ we have $\stlang{z} = \Sigma^* = \$^{-1}\lang(\aut')$ and $\stlang{x} = \#^{-1}\lang(\aut')$ (see Figure~\ref{fig:proof-undecidability-residuality}).
The only remaining state for which we do not yet know whether the residuality property holds is state $y$.

If $\lang(\aut) = \Sigma^*$ (i.e., the original automaton is universal), then we note that $\stlang{y} = \stlang{x}$.
In this case, $\stlang{y} = \#^{-1}\lang(\aut')$.
So, in this case, $\aut'$ is residual.

Suppose that $\aut'$ is residual.
Then $\stlang{y} = w^{-1}\lang'$ for some word $w$.
Provided that $\lang(\aut)$ is not empty, there is some $u \in \lang(\aut)$.
So we know that $\$u \in \stlang{y}$.
This means that word $w$ cannot start with $a \in \Sigma$, $q$, $\underline{q}$ for $q \in Q$, or $\$$ as their derivatives do not contain $\$u$.
The only possibility is that $w = \#^k$ for some $k > 0$.
This implies $\stlang{y} = \stlang{x}$,
meaning that the language of $\aut$ is universal.

This proves that $\aut$ is universal iff $\aut'$ is residual.
\end{proof}

These results also hold for the subclass of non-guessing automata, as the constructions do not introduce any guessing and universality for non-guessing nondeterministic nominal automata is undecidable.

\paragraph*{Closure properties}
We will now show that several closure properties fail for residual languages. Interestingly, this parallels the situation for probabilistic languages: residual ones are not even closed under convex sums.
We emphasise that residual automata were devised for learning purposes, where closure properties play no significant role. In fact, one typically exploits closure properties of the wider class of nondeterministic models, e.g., for automata-based verification.
The following results show that in our setting this is indeed unavoidable.

Consider the alphabet $\Sigma = \atoms \cup \{ \Anchor(a) \mid a \in \atoms \}$ and the residual language $\lang_\up{r}$ from Section~\ref{sec:examples}.
We consider a second language $\lang_2 = \atoms^{*}$ which can be accepted by a deterministic (hence residual) automaton. We have the following non-closure results:
\begin{description}
\item[Union]
The language $\lang = \lang_\up{r} \cup \lang_2$ cannot be accepted by a residual automaton.
In fact, although derivatives of the form $\Anchor(a)^{-1}\lang$ are still join-irreducible (see Section~\ref{sec:examples}, residual case),
they
have no summand $\atoms^*$, which means that they cannot generate $a^{-1}\lang = \atoms^{*} \cup \bigcup_{b \neq a} \Anchor(b)^{-1}\lang$. By Theorem~\ref{thm:residual-characterisation}(3) it follows that $\lang$ is not residual.

\item[Intersection]
The language $\lang = \lang_\up{r} \cap \lang_2 = \lang_\up{n}$ cannot be accepted by a residual automaton, as we have seen in Section~\ref{sec:examples}.

\item[Reversal]
The language $\{ a w \mid a \text{ not in } w \}$ is residual (even deterministic), but its reverse language is $\lang_\up{n}$ and cannot be accepted by a residual automaton.

\item[Complement]
Consider the language $\lang_\up{ng,r}$ of words where some atom occurs twice. Its complement $\overline{\lang_\up{ng,r}}$ is the language of all fresh atoms, which cannot even be recognised by a nondeterministic nominal automaton~\cite{BojanczykKL14}.
\end{description}
Closure under \ch{concatenation and} Kleene star is yet to be settled.

\subsection{Length of characterising words}
We end this section by giving a result about the length of characterising words.
Note that in the finite case, the characterising words of an $n$-state residual automaton have length at most $2^n$, since one can determinise automata.
In our case, this no longer holds, and we show that the length of characterising word is not bounded in the number of states only. \ch{We state this result in terms of register automata to help intuition.}

\begin{proposition}
There is a family of residual register automata $\aut_k$ $(k \geq 1)$ with two states and $k$ registers of which the characterising words have length $k$.
\end{proposition}
\begin{proof}
We define a variation on the automaton $\aut_\up{r}$ from Section~\ref{sec:examples} using the alphabet $\Sigma = \atoms \cup \{ \Anchor(a) \mid a \in \atoms \}$.
The automaton $\aut_k$ is defined by the following sets, where ${\atoms \choose k}$ denotes the set of $k$-element subsets of $\atoms$ (\ch{note that $\aut_1 = \aut_r$}):
\[ Q = {\atoms \choose k} \cup \{\top\} \qquad I = {\atoms \choose k} \qquad F = \{\top\} \]
\begin{align*}
\delta = &\left\{ (S, \Anchor(a), S) \,\middle|\, S \in \textstyle {\atoms \choose k}, a \in S \right\} \\
          &\cup \left\{ (S, a, S) \,\middle|\, S \in \textstyle {\atoms \choose k}, a \notin S \right\} \\
          &\cup \left\{ (S, a, \top) \,\middle|\, S \in \textstyle {\atoms \choose k}, a \in S \right\}
\end{align*}
A state $S = \{ a_1, \ldots, a_k \} \in {\atoms \choose k}$ is anchored by the word $w = \Anchor(a_1) \ldots \Anchor(a_k)$, which is of length $k$.
This is also the shortest characterising word for that state.
Note that $Q$ only has two orbits, meaning that a register automaton equivalent to this nominal automaton only requires two states.
\end{proof}

\section{Exact learning}
\label{sec:learning}

In our previous paper on learning nominal automata~\cite{MoermanS0KS17}, we provided a learning algorithm to learn residual automata, that converges for deterministic languages.
However, we observed by experimentation that the algorithm was also able to learn certain nondeterministic languages.
At that point we did not know which class of languages could be accepted by residual nominal automata, and so it was left open whether the algorithm converges for all residual languages.
In this section we will answer this question negatively, but also provide a modified algorithm which does always converge.

\subsection{Angluin-style learning}

We briefly review the classical automata learning algorithms \lstar\ by Angluin~\cite{Angluin87} for deterministic automata, and \nlstar\ by Bollig et al.~\cite{BolligHKL09} for residual automata.

Both algorithms can be seen as a game between two players: \emph{the learner} and \emph{the teacher}.
The learner aims to construct the minimal automaton for an unknown language $\lang$ over a finite alphabet $\Sigma$.
In order to do this, it may ask the teacher, who knows about the language, two types of queries:
\begin{description}
\item[Membership query] Is a given word $w$ in the target language, i.e., $w \in \lang$?
\item[Equivalence query] Does a given \emph{hypothesis} automaton $\mathcal{H}$ recognise the target language, that is, is $\lang = \lang(\mathcal{H})$?
\end{description}
If the teacher replies \emph{yes} to an equivalence query, then the algorithm terminates, as the hypothesis $\mathcal{H}$ is correct. Otherwise, the teacher must supply a \emph{counterexample}, that is a word in the symmetric difference of $\lang$ and $\lang(\mathcal{H})$.
Availability of equivalence queries may seem like a strong assumption and in fact it is often weakened by allowing only random sampling (see~\cite{KearnsV94} or~\cite{Vaandrager17} for details).

Observations about the language made by the learner via queries are stored in an \emph{observation table} $\ot$. This is a table where rows and columns range over two finite sets of words $S,E \subseteq \Sigma^\star$ respectively, and $\ot(u,v) = 1$ if and only if $uv \in \lang$.
Intuitively, each row of $\ot$ approximates a derivative of $\lang$, in fact we have $\ot(u) \subseteq u^{-1}\lang$. However, the information contained in $\ot$ may be incomplete:
some derivatives $w^{-1}\lang$ are not reached yet because no membership queries for $w$ have been posed, and some pairs of rows $\ot(u)$, $\ot(v)$ may seem equal to the learner, because no word has been seen yet which distinguishes them.
The learning algorithm will add new words to $S$ when new derivatives are discovered, and to $E$ when words distinguishing two previously identical derivatives are discovered.

The table $\ot$ is \emph{closed} whenever one-letter extensions of derivatives are already in the table, i.e., $\ot$ has a row for $ua^{-1}\lang$, for all $u \in S,a \in \Sigma$.
If the table is closed,\footnote{\lstar\ also needs the table to be \emph{consistent}. We do not need that in our discussion here.}
\lstar\ is able to construct an automaton from $\ot$, where states are distinct rows (i.e., derivatives). The construction follows the classical one for the canonical automaton of a language from its derivatives~\cite{Nerode58}.
The \nlstar\ algorithm uses a modified notion of closedness, where one is allowed to take unions (i.e., a one-letter extension can be written as unions of rows in $\ot$), and hence is able to learn a RFSA accepting the target language.

When the table is not closed, then a derivative is missing, and a corresponding row needs to be added. Once an automaton is constructed, it is submitted in an equivalence query. If a counterexample is returned, then again the table is extended, after which the process is repeated iteratively.
The \lstar{} and \nlstar{} algorithms adopt different counterexample-handling strategies: the former adds a new row, the latter a new column.
Both result in a new derivative being detected.

\subsection{The nominal case}

In~\cite{MoermanS0KS17} we have given nominal versions of $\lstar$ and $\nlstar$, called $\nomlstar{}$ and $\nomnlstar{}$ respectively.
They seamlessly extend the original algorithms by operating on orbit-finite sets.
This allows us to learn automata over infinite alphabets, but using only finitely many queries.
The algorithm $\nomlstar$ always terminates for deterministic languages, because the language only has orbit-finitely many distinct derivatives (Theorem~\ref{thm:my-ne}), and hence only need orbit-finitely many distinct rows in the observation table.
However, it will \emph{never} terminate for languages not accepted by deterministic automata (such as residual or nondeterministic languages).
\begin{theorem}[\cite{Moerman19}]
$\nomlstar$ converges if and only if $\Res(\lang)$ is orbit-finite, in which case it outputs the canonical deterministic automaton accepting $\lang$.
Moreover, at most $\mathcal{O}(nk)$ equivalence queries are needed, where $n$ is the number of orbits of the minimal deterministic automaton, and $k$ is the maximum support size of its states.
\end{theorem}
The nondeterministic case is more interesting.
Using Theorem~\ref{thm:residual-characterisation}, we can finally establish which nondeterministic languages can be characterised via orbit-finitely many observations.
\begin{corollary}[of Theorem~\ref{thm:residual-characterisation}]
Let $\lang$ be a nondeterministic nominal language.
If $\lang$ is a residual language,
then there exists an observation table with orbit-finitely many rows and columns from which we can construct the canonical residual automaton.
\end{corollary}
This explains why in~\cite{MoermanS0KS17} \nomnlstar{} was able to learn some residual nondeterministic automata: an orbit-finite observation table exists, which allows \nomnlstar{} to construct the canonical residual automaton.
Unfortunately, the \nomnlstar{} algorithm does not guarantee that it always finds this orbit-finite observation table.
We only have that guarantee for deterministic languages.
The following example shows that \nomnlstar{} may indeed diverge when trying to close the table.

\begin{example}
Suppose \nomnlstar{} tries to learn the residual language $\lang$ accepted by the automaton below over the alphabet $\Sigma = \atoms \cup \{ \Anchor(a) \mid a \in \atoms \}$. This is a slight modification of the residual language of Section~\ref{sec:examples}.

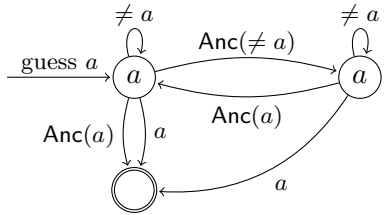
\begin{center}
\begin{tikzpicture}[automaton]
\node[state] (q0) at (0,0) {$a$};
\node[state] (q1) at (2,0) {$a$};
\node[state, accepting] (fin) at (0,-1) {};

\path[tr]
(init) edge node {guess $a$} (q0)
(q0) edge[bend left=15] node {$\Anchor(\neq a)$} (q1)
(q0) edge[loop above] node {$\neq a$} (q0)
(q1) edge[loop above] node {$\neq a$} (q1)
(q1) edge[bend left=15] node {$\Anchor(a)$} (q0)
(q1) edge[bend left] node {$a$} (fin)
(q0) edge[bend right=15] node[left] {$\Anchor(a)$} (fin)
(q0) edge[bend left=15] node {$a$} (fin);
\end{tikzpicture}
\end{center}
The algorithm starts by considering the row for the empty word $\epsilon$, and its one-letter extensions $\epsilon \cdot a = a$ and $\epsilon \cdot \Anchor(a) = \Anchor(a)$. These rows correspond to the derivatives $\epsilon^{-1}\lang = \lang$, $a^{-1}\lang$ and $\Anchor(a)^{-1}\lang $.
Column labels are initialised to the empty word $\epsilon$. At this point $a^{-1} \lang$ and $\Anchor(a)^{-1} \lang$ appear identical, as the only column $\epsilon$ does not distinguish them.
However, they appear different from $\epsilon^{-1}\lang$, so the algorithm will add the row for either $a$ or $\Anchor(a)$ in order to close the table.
Suppose the algorithm decides to add $a$. Then it will consider one-letter extensions $ab$, $abc$, $abcd$, etc.
Since these correspond to different derivatives --- each strictly smaller than the previous one --- the algorithm will get stuck in an attempt to close the table.
At no point it will try to close the table with the word $\Anchor(a)$, since it stays equivalent to $a$.
So in this case \nomnlstar{} will not terminate.
However, if the algorithm instead adds $\Anchor(a)$ to the row labels, it will then also add $\Anchor(a)\Anchor(b)$, which is a characterising word for the initial state.
In that case, \nomnlstar{} will terminate.
\end{example}

\subsection{Modified \nomnlstar{}}

We modify the \nomnlstar{} algorithm from~\cite{MoermanS0KS17} to ensure that it always terminates.
We do this by changing how the table will be closed.
The algorithm is shown in Algorithm~\ref{fig:nfa-alg} with the changes made to \nomnlstar{} in red.
In short, the change is as follows.
When the algorithm adds a word $w$ to the set of rows, then it also adds all other words of length $|w|$.
Since all words of bounded length are added, the algorithm will eventually find all characterising words of the canonical residual automaton, and it will therefore be able to reconstruct this automaton.

\begin{algorithm}
\begin{codebox}
	\Procname{\proc{Modified \nomnlstar{} learner}}
	\li Initialise $\ot$ with $S, E \gets \{\epsilon\}$
	\li \Repeat
	\li \While $\ot$ is not join-closed or not join-consistent \Do
	\li \If $\ot$ is not join-closed
	\li \Then\label{line:begin-nfa-closed}find $s\in S,a \in A$ such that
	$\row(sa) \in \JI(\Rows(\ot)) \setminus \RowsUpp(\ot)$
	\li {\color{red}$l \gets$ length of the word $sa$}
	\li {\color{red}$S \gets S \cup \Sigma^{\leq l}$}\label{line:clos-add-row}
	\End\label{line:nfa-end-closed}
	\li \If $\ot$ is not join-consistent\label{line:nfa-begin-const}
	\li \Then find $s_1, s_2 \in S$, $a \in A$, and $e \in E$ such that \zi \qquad $\row(s_1) \leq \row(s_2)$ but $e \in \row(s_1 a)$ and $e \notin \row(s_2 a)$
	\label{line:nfa-cons-witness}
	\li $E \gets E \cup \orb(ae)$
	\label{line:cons-add-col}
	\End
	\End
	\label{line:nfa-end-const}
	\li Query $\mathcal{H} \gets \aut(\ot)$ for equivalence
	\label{line:nfa-conj}
	\li \If the Teacher replies \textbf{no}, with a counter-example $t$
	\li \Then $E \gets E \cup \{ \orb(t_0) \mid t_0 \text{ is a suffix of } t \}$
	\End
	\li \Until the Teacher replies \textbf{yes} to the equivalence query.
	\li \Return $\mathcal{H}$
\end{codebox}
\caption{Modified nominal \nlstar{} algorithm for Theorem~\ref{thm:learning-residual}.}\label{fig:nfa-alg}
\end{algorithm}

We briefly recall the notation we use in the algorithm and afterwards prove convergence.
We denote the observation table $\ot$ by the pair $\ot = (S, E)$ of row and column indices.
The membership will be queried for the set $S E \cup S \Sigma E$ and these observations will be stored during the algorithm.
The function $\row_{\ot} \colon S \cup S \Sigma \to \Powfs(E)$ returns the content of each row, i.e., $\row_{\ot}(t) \coloneqq \{ e \in E \mid te \in \lang \}$.%
\footnote{It is more common to use $2^E$ instead of $\Powfs(E)$ and use a more ``functional'' approach in which $\row_{\ot}(s)(e) = \lang(se)$. However, we use $\Powfs(E)$ to remain consistent with our notation in the paper.}
Since $\row_{\ot}$ takes values in a nominal join-semilattice, we use the notation (${\leq}, {<}, {\vee}, \ldots$) from Section~\ref{sec:nominal-join-semilattices} on rows.
Note that $\row_{\ot}(s) = s^{-1} \lang \cap E$ can be thought of an approximation to $s^{-1} \lang$.
We will omit the subscript $T$ in $\row$.

Given an observation table $\ot = (S, E)$, we define the set of rows as
\[ \Rows(\ot) \coloneqq \{ \row(t) \mid t \in S \cup S\Sigma \} . \]
This is an orbit-finite poset, ordered by $\leq$, that is, the order on $\Powfs(E)$ given by subset inclusion.
We define the set of \emph{upper rows} as $\RowsUpp(T) \coloneqq \{ \row(s) \mid s \in S \} \subseteq \Rows(\ot)$.

\begin{definition}
A table $\ot = (S, E)$ is
\begin{itemize}
\item \emph{join-closed} if for each $s \in S$ and $a \in \Sigma$ we have
\[ \row(sa) = \bigvee \{ \row(s) \mid \row(s) \leq \row(sa), s \in \JI(\Rows(\ot)) \cap \RowsUpp(\ot) \}, \]
in words: each extended row $sa$ can be obtained as a join of join-irreducible rows in $S$;
\item \emph{join-consistent} if for all $s_1, s_2 \in S$ and $a \in \Sigma$ we have
\[ \row(s_1) \leq \row(s_2) \quad \implies \quad \row(s_1 a) \leq \row(s_2 a) . \]
\end{itemize}
\end{definition}

Another way to define join-closedness would be to consider the set $\JI(\RowsUpp(\ot))$ instead of $\JI(\Rows(\ot)) \cap \RowsUpp(\ot)$.
This would slightly change the algorithm, but not substantially.
We stick to the original description of \nlstar{}~\cite{BolligHKL09}.

\begin{construction}
Given a join-closed and join-consistent observation table $\ot = (S, E)$ we define an automaton $\aut(\ot) = (\Sigma, Q, I, F, \delta)$ as follows:
\begin{align*}
Q &= \JI(\Rows(T)) \cap \RowsUpp(T) \\
I &= \{ r \in Q \mid r \leq \row(\epsilon) \} \\
F &= \{ r \in Q \mid \epsilon \in r \} \\
\delta &= \{ (\row(s), a, r') \mid s \in S, r' \in Q, r' \leq \row(s a) \}
\end{align*}
This closely follows the definition of the canonical residual automaton in Theorem~\ref{thm:residual-characterisation}, but only uses the information from the observation table.
This construction is effective, because one can decide whether $r \in \Rows(T)$ is a join-irreducible element if $r \neq \bigvee \{y \in \Rows(T) \mid y < x \}$ and $r$ is non-empty.
This is a set-builder expression in the programming language developed in~\cite{Bojanczyk19}.
The rest of the construction is directly given as set-builder expressions.
\end{construction}

In the remainder of this section we prove termination of our modified learning algorithm.
In the following, $|X|$ is the \emph{orbit-count} of an orbit-finite set $X$.
By \emph{atom-dimension} of $X$ we mean the maximal size of supports of elements of $X$.
Let $p(k)$ denote the number of orbits of the set $\atoms^{(k)} \times \atoms^{(k)}$, where $\atoms^{(k)} = \{ (a_1,a_2,\dots,a_k) \mid a_i \in \atoms, \forall i,j: a_i \neq a_j \}$;
it equals the number of partial permutations on an $k$-element set.

\begin{theorem}
\label{thm:learning-residual}
Algorithm~\ref{fig:nfa-alg} query learns residual nominal languages.
Moreover, it uses at most $\mathcal{O}(l + |\Sigma^{\leq l}|^2 \cdot p(dl))$ equivalence queries,
where $l$ is the length of the longest characterising word and $d$ is the atom-dimension of $\Sigma$.
\end{theorem}

The theorem will be proven with the help of the following lemmata.

\begin{lemma}\label{lem:consistent}
\ch{
Every observation table $(S,E)$ during the execution can be extended to a join-consistent table $(S, E')$ by adding orbit-finitely many columns.
}
\end{lemma}
\begin{proof}
Note that the $\row$ function defines a preorder $\sqsubseteq$ on $S$ defined by $s_1 \sqsubseteq s_2$ iff $\row(s_1) \leq \row(s_2)$.
Each time columns are added \ch{in line~\ref{line:cons-add-col}} to solve join-inconsistency, this preorder is refined.
So we obtain a chain of preorders:
\[ \cdots \,\subsetneq\, {\sqsubseteq_3} \,\subsetneq\, {\sqsubseteq_2} \,\subsetneq\, {\sqsubseteq_1} \,\subseteq\, {S \times S} . \]
\ch{When the preorder has been maximally refined it will only contain identity pairs, so join-consistency trivially holds.}
Since the set $S \times S$ is orbit-finite, this chain has a length at most $|S \times S|$.
The number of orbits in $S \times S$ is bounded\footnote{Slightly better bounds can be given by triangular numbers, but the asymptotics remain the same.}
  by $|S|^2 \cdot p(k)$, where $k$ is the atom-dimension of $S$.%
\end{proof}

\begin{lemma}\label{lem:closed}
If the set $S$ of an observation table $T = (S, E)$ contains all characterising words, then the table is join-closed.
\end{lemma}
\begin{proof}
For this lemma, we consider the idealised observation table $\ot' = (S, \Sigma^*)$.
In this table, we have $\row_{\ot'}(s) = s^{-1}\lang$ for each $s \in S$.
Since $S$ contains all characterising words, we have $\JI(\Rows(\ot')) = \JI(\Res(\lang))$.
This means that the idealised table is join-closed, that is, for $s \in S$ and $a \in \Sigma$ we have
\[ \row_{\ot'}(sa) = \bigvee_{s \in I} \row_{\ot'}(s), \]
for a suitable orbit-finite set of row indices $I$.
This equation still holds if we restrict to the set $E$,
and so the table $\ot$ is join-closed.
\end{proof}

\begin{lemma}\label{lem:closed-consistent}
Given an observation table $(S, E)$ for a residual language, the algorithm extends it to a join-closed and join-consistent table $(S', E')$ in finitely many steps.
\end{lemma}
\begin{proof}
\ch{We will show that lines \ref{line:clos-add-row} and \ref{line:cons-add-col} are executed finitely many times.
Line \ref{line:clos-add-row} adds $\Sigma^l$ incrementally to the set $S$ with increasing $l$, and $l$ is only increased until $\Sigma^{\leq l}$ contains all the required characterising words (Lemma~\ref{lem:closed}).
When line \ref{line:cons-add-col} is executed, we have two cases: if the resulting $(S,E')$ table is join-closed but not join-consistent, we keep adding columns, but this can only happen finitely-many times (Lemma~\ref{lem:consistent}); otherwise $S$ is extended, but as previously shown this can only happen finitely many times.}
\end{proof}

\begin{lemma}\label{lem:agree}
The constructed hypothesis agrees with the values in the table.
\end{lemma}
\begin{proof}
This closely follows the inductive proof of Theorem~\ref{thm:residual-characterisation}.
Recall that the states of the automaton are given by $q = \row(s) \subseteq E$ for certain $s$.
We prove $\stlang{\row(s)} \cap E = \row(s)$ for states $row(s)$ by induction:
\begin{align*}
\epsilon \in \stlang{\row(s)} \cap E &\,\iff\, \row(s) \in F \text{ and } \epsilon \in E \,\iff\, \epsilon \in \row(s) \wedge \epsilon \in E \\
au \in \stlang{\row(s)} \cap E
&\,\iff\, u \in \stlang{\delta(\row(s), a)} \\
&\,\iff\, \exists s' \text{ with } \row(s') \leq \row(sa) \text{ and } u \in \stlang{\row(s')} \cap E \\
&\,\iff\, \exists s' \text{ with } \row(s') \leq \row(sa) \text{ and } u \in \row(s') \text{ and } u \in E \\
&\,\iff\, u \in \bigvee \delta(\row(s), a) \text{ and } u \in E \\
&\,\iff\, au \in \row(s) \text{ and } u \in E
\end{align*}
Note that we need $E$ to be suffix-closed and that the empty word is in $E$.
\end{proof}

We can now prove the main theorem of this section.

\begin{proof}[Proof of Theorem~\ref{thm:learning-residual}]
Each time a counterexample is added, the next hypothesis (which will always be constructed per Lemma~\ref{lem:closed-consistent}) will be different (Lemma~\ref{lem:agree}).
But this can only happen if a column or row is added, which we only need to do finitely many times (Lemmas~\ref{lem:closed} and~\ref{lem:consistent}).
To be precise, this happens at most
\[ l + |\Sigma^{\leq l}|^2 \cdot p(dl) \]
times, where $p$ is from the proof of Lemma~\ref{lem:consistent} and $d$ is the atom-dimension of $\Sigma$ and $l$ is the least such that $\Sigma^{\leq l}$ contains a characterising word for each state of the canonical residual automaton.
(Put differently: consider the shortest characterising words for all states, then $l$ is the length of the longest of these.)
So we conclude that the algorithm terminates after finitely many equivalence queries and only need finitely many membership queries since the table $(S \cup S \Sigma) \times E$ is orbit-finite.
\end{proof}

Unfortunately, considering all words bounded by a certain length requires many membership queries.
In fact, characterising words can be exponential in length~\cite{DenisLT02}, meaning that this algorithm may need doubly exponentially many membership queries.%
\footnote{The reader should not interpret this as a complexity upper bound. In fact, no upper bound is known on the length of characterising words.}

\begin{remark}
Note that our termination argument is not concerned with the implementation of the teacher.
This is standard for Angluin-style algorithms, which assume that the teacher is always able to provide correct answers to queries.
As mentioned, this assumption is often too strong, and in our setting a direct equivalence check is not available due to Proposition~\ref{prop:decidability-universality}.
In practice, however, it is common to use testing techniques~\cite{Vaandrager17}.
\end{remark}

\section{Discussion}

\subsection{Conclusion}

In this paper we have investigated a subclass of nondeterministic automata over infinite alphabets.
This class naturally arises in the context of query learning, where automata have to be constructed from finitely many observations.
Although there are many classes of data languages, we have shown that our class of residual languages admits canonical automata.
The states of these automata correspond to join-irreducible elements.

In the context of learning, we show that convergence of standard Angluin-style algorithms is not guaranteed, even for residual languages.
We propose a modified algorithm which guarantees convergence at the expense of an increase in the number of observations.

We emphasise that, unlike other algorithms based on residuality such as \nlstar~\cite{BolligHKL09} and \alstar~\cite{AngluinEF15}, our algorithm does not depend on the size, or even the existence, of the minimal deterministic automaton for the target language.
This is a crucial difference, since dependence on the minimal deterministic automaton hinders generalisation to nondeterministic nominal automata, which are strictly more expressive.
Ideally, in the residual case, one would like to have an efficient algorithm for which the complexity depends only on the length of characterising words, which is an intrinsic feature of residual automata.
To the best of our knowledge, no such algorithm exists in the finite setting.

Finally, another interesting open question is whether all nondeterministic automata can be \emph{efficiently} learned.
We note that nondeterministic automata can be enumerated, and hence can be learned via equivalence queries only.
This would result in a highly inefficient algorithm.
This parallels the current understanding of learning probabilistic languages. Although efficient (learning in the limit) learning algorithms for deterministic and residual languages exist~\cite{DenisE04}, the general case is still open.

\subsection{Related Work} We last present some related work.
\paragraph*{Lattices and Category Theory}
In~\cite{GabbayG17, GabbayLP11} aspects of nominal lattices and Boolean algebras are investigated.
To the best of our knowledge, our results of nominal lattice theory, especially the algorithmic properties (of join-irreducible elements), are new.

Residual automata over finite alphabets have categorical characterisation~\cite{MyersAMU15} in terms of \emph{closure spaces}.
We see no obstructions in generalising those results to nominal sets.
This would amount to finding the right notion of nominal (complete) join-semilattice, with either finitely or uniformly supported joins, and the notion of nominal closure spaces.

\paragraph*{Other Data Languages}
Related data languages are the nominal languages with an explicit notion of binding~\cite{GabbayC11, KozenMP015, KurzST12, SchroderKMW17}.
Although these are sub-classes of the nominal languages we consider, binding is an important construct, e.g., to represent resource-allocation.
Availability of a notion of derivatives~\cite{KozenMP015} suggests that residuality may prove beneficial for learning these languages.

Another related type of automaton is that of session automata~\cite{BolligHLM14}.
They differ in that they deal with global freshness as opposed to local freshness.
They form a robust class of languages and their learnability is discussed in \emph{loc.\ cit.}

\ch{Nominal} automata can be defined parametrically in the data domain (see Section~\ref{sec:symmetries}, but also~\cite{Bojanczyk19}).
For the mathematical foundation of nominal sets (or sets with atoms), the only difference will be that of the group of symmetries.
However, the current proof of Lemma~\ref{lem:join-of-ji} only works when the permutations are finite.
This excludes the ordered atoms with monotone bijections as permutations.
We do not know whether the main theorem still holds in the general case.

\paragraph*{Alternating Automata}
One could try to generalise \alstar{} from~\cite{AngluinEF15} to \emph{alternating \ch{nominal} automata.}
Beside the join, these automata can also use a meet and so the algebraic structure will be that of a distributive (nominal) lattice.

We think that the analogue to Theorem~\ref{thm:residual-characterisation} will hold: a language is accepted by a residual alternating nominal automaton iff there is a orbit-finite set of generators w.r.t.\ both the join and meet.
However, there is no analogue to the join-irreducibles: there can be many different sets of generators.
This might be an obstacle to generalise \alstar{} as the algorithm will need to find a set of generating rows, and we are unsure whether this is even decidable.

\paragraph*{Unambiguous Automata}
Of special interest is the subclass of \emph{unambiguous automata} which enjoy many recent breakthroughs~\cite{BarloyC21, Colcombet15, MottetQ19}.
We note that residual languages are orthogonal to unambiguous languages.
For instance, the language $\lang_\up{n}$ is unambiguous but not residual, whereas $\lang_\up{ng,r}$ is residual but ambiguous.
Moreover, their intersection has neither property, and every deterministic language has both properties.
One interesting fact is that if a canonical residual automaton is unambiguous, then the join-irreducibles form an anti-chain.

The unambiguous automata can be embedded into \emph{weighted nominal automata}~\cite{BojanczykKM21}.
This embedding allows one to use linear algebra for these automata and
we expect that the learning algorithm for weighted automata~\cite{BalleM15, BergadanoV96} generalise to the setting of nominal automata.
This could follow from the algebraic learning results of~\cite{UrbatS20}, since the length of the nominal vector spaces is finite~\cite{BojanczykKM21}.
The learning algorithm can therefore construct a minimal weighted nominal automaton, but this might not be an actual unambiguous (nondeterministic) automaton.

\section*{Acknowledgements}

We would like to thank Gerco van Heerdt for providing examples similar to that of $\lang_\up{r}$ in the context of probabilistic automata. We thank Borja Balle for references on residual probabilistic languages, and Henning Urbat for discussions on nominal lattice theory. We thank Thorsten Wi{\ss}mann for his detailed comments that led us to simplify and improve several proofs, and the overall presentation. We thank reviewers for their interesting questions and suggestions.

\bibliographystyle{alphaurl}
\bibliography{references}

\newcommand{\etalchar}[1]{$^{#1}$}
\begin{thebibliography}{MAMU15}

\bibitem[AEF15]{AngluinEF15}
Dana Angluin, Sarah Eisenstat, and Dana Fisman.
\newblock Learning regular languages via alternating automata.
\newblock In {\em {IJCAI}}, pages 3308--3314. {AAAI} Press, 2015.

\bibitem[Ang87]{Angluin87}
Dana Angluin.
\newblock Learning regular sets from queries and counterexamples.
\newblock {\em Inf. Comput.}, 75(2):87--106, 1987.

\bibitem[BBKL12]{BojanczykBKL12}
Miko{\l}aj Boja{\'{n}}czyk, Laurent Braud, Bartek Klin, and Slawomir Lasota.
\newblock Towards nominal computation.
\newblock In {\em {POPL}}, pages 401--412, 2012.
\newblock \href {https://doi.org/10.1145/2103656.2103704}
  {\path{doi:10.1145/2103656.2103704}}.

\bibitem[BC21]{BarloyC21}
Corentin Barloy and Lorenzo Clemente.
\newblock Bidimensional linear recursive sequences and universality of
  unambiguous register automata.
\newblock In {\em {STACS}}, volume 187 of {\em LIPIcs}, pages 8:1--8:15.
  Schloss Dagstuhl - Leibniz-Zentrum f{\"{u}}r Informatik, 2021.

\bibitem[BHKL09]{BolligHKL09}
Benedikt Bollig, Peter Habermehl, Carsten Kern, and Martin Leucker.
\newblock Angluin-style learning of {NFA}.
\newblock In {\em {IJCAI}}, pages 1004--1009, 2009.

\bibitem[BHLM14]{BolligHLM14}
Benedikt Bollig, Peter Habermehl, Martin Leucker, and Benjamin Monmege.
\newblock A robust class of data languages and an application to learning.
\newblock {\em Logical Methods in Computer Science}, 10(4), 2014.
\newblock \href {https://doi.org/10.2168/LMCS-10(4:19)2014}
  {\path{doi:10.2168/LMCS-10(4:19)2014}}.

\bibitem[BKL14]{BojanczykKL14}
Miko{\l}aj Boja{\'{n}}czyk, Bartek Klin, and Slawomir Lasota.
\newblock Automata theory in nominal sets.
\newblock {\em Logical Methods in Computer Science}, 10(3), 2014.

\bibitem[BKM21]{BojanczykKM21}
Miko{\l}aj Boja{\'{n}}czyk, Bartek Klin, and Joshua Moerman.
\newblock Orbit-finite-dimensional vector spaces and weighted register
  automata.
\newblock In {\em {LICS}}, 2021.
\newblock To appear.

\bibitem[BLLR17]{BerndtLLR17}
Sebastian Berndt, Maciej Li{\'{s}}kiewicz, Matthias Lutter, and R{\"{u}}diger
  Reischuk.
\newblock Learning residual alternating automata.
\newblock In {\em {AAAI}}, pages 1749--1755, 2017.
\newblock URL:
  \url{http://aaai.org/ocs/index.php/AAAI/AAAI17/paper/view/14748}.

\bibitem[BM15]{BalleM15}
Borja Balle and Mehryar Mohri.
\newblock Learning weighted automata.
\newblock In {\em {CAI}}, volume 9270 of {\em Lecture Notes in Computer
  Science}, pages 1--21. Springer, 2015.

\bibitem[Boj19]{Bojanczyk19}
Miko{\l}aj Boja{\'{n}}czyk.
\newblock {\em Slightly Infinite Sets}.
\newblock Draft September 6, 2019, 2019.
\newblock URL: \url{https://www.mimuw.edu.pl/~bojan/paper/atom-book}.

\bibitem[BT18]{BojanczykT18}
Miko{\l}aj Boja{\'{n}}czyk and Szymon Toru{\'{n}}czyk.
\newblock On computability and tractability for infinite sets.
\newblock In {\em {LICS}}, pages 145--154. {ACM}, 2018.

\bibitem[BV96]{BergadanoV96}
Francesco Bergadano and Stefano Varricchio.
\newblock Learning behaviors of automata from multiplicity and equivalence
  queries.
\newblock {\em {SIAM} J. Comput.}, 25(6):1268--1280, 1996.

\bibitem[Col15]{Colcombet15}
Thomas Colcombet.
\newblock Unambiguity in automata theory.
\newblock In {\em Descriptional Complexity of Formal Systems - 17th
  International Workshop, {DCFS} 2015, Waterloo, ON, Canada, June 25-27, 2015.
  Proceedings}, pages 3--18, 2015.
\newblock \href {https://doi.org/10.1007/978-3-319-19225-3_1}
  {\path{doi:10.1007/978-3-319-19225-3_1}}.

\bibitem[DE04]{DenisE04}
Fran{\c{c}}ois Denis and Yann Esposito.
\newblock Learning classes of probabilistic automata.
\newblock In {\em {COLT}}, volume 3120 of {\em Lecture Notes in Computer
  Science}, pages 124--139. Springer, 2004.

\bibitem[DE08]{DenisE08}
Fran{\c{c}}ois Denis and Yann Esposito.
\newblock On rational stochastic languages.
\newblock {\em Fundam. Inform.}, 86(1-2):41--77, 2008.

\bibitem[DLT02]{DenisLT02}
Fran{\c{c}}ois Denis, Aur{\'{e}}lien Lemay, and Alain Terlutte.
\newblock Residual finite state automata.
\newblock {\em Fundam. Inform.}, 51(4):339--368, 2002.

\bibitem[DP02]{DaveyP02}
Brian~A. Davey and Hilary~A. Priestley.
\newblock {\em Introduction to Lattices and Order, Second Edition}.
\newblock Cambridge University Press, 2002.
\newblock \href {https://doi.org/10.1017/CBO9780511809088}
  {\path{doi:10.1017/CBO9780511809088}}.

\bibitem[GC11]{GabbayC11}
Murdoch~James Gabbay and Vincenzo Ciancia.
\newblock Freshness and name-restriction in sets of traces with names.
\newblock In {\em {FOSSACS}s}, pages 365--380, 2011.
\newblock \href {https://doi.org/10.1007/978-3-642-19805-2_25}
  {\path{doi:10.1007/978-3-642-19805-2_25}}.

\bibitem[GDPT13]{GrigoreDPT13}
Radu Grigore, Dino Distefano, Rasmus~Lerchedahl Petersen, and Nikos Tzevelekos.
\newblock Runtime verification based on register automata.
\newblock In {\em {TACAS}}, volume 7795 of {\em Lecture Notes in Computer
  Science}, pages 260--276. Springer, 2013.

\bibitem[GG17]{GabbayG17}
Murdoch~James Gabbay and Michael Gabbay.
\newblock Representation and duality of the untyped {\(\lambda\)}-calculus in
  nominal lattice and topological semantics, with a proof of topological
  completeness.
\newblock {\em Ann. Pure Appl. Logic}, 168(3):501--621, 2017.
\newblock \href {https://doi.org/10.1016/j.apal.2016.10.001}
  {\path{doi:10.1016/j.apal.2016.10.001}}.

\bibitem[GLP11]{GabbayLP11}
Murdoch~James Gabbay, Tadeusz Litak, and Daniela Petrișan.
\newblock Stone duality for nominal boolean algebras with {N}.
\newblock In {\em {CALCO}}, volume 6859 of {\em Lecture Notes in Computer
  Science}, pages 192--207. Springer, 2011.

\bibitem[HJV19]{HowarJV19}
Falk Howar, Bengt Jonsson, and Frits~W. Vaandrager.
\newblock Combining black-box and white-box techniques for learning register
  automata.
\newblock In {\em Computing and Software Science}, volume 10000 of {\em Lecture
  Notes in Computer Science}, pages 563--588. Springer, 2019.

\bibitem[KF94]{KaminskiF94}
Michael Kaminski and Nissim Francez.
\newblock Finite-memory automata.
\newblock {\em Theor. Comput. Sci.}, 134(2):329--363, 1994.

\bibitem[KMPS15]{KozenMP015}
Dexter Kozen, Konstantinos Mamouras, Daniela Petrișan, and Alexandra Silva.
\newblock Nominal kleene coalgebra.
\newblock In {\em {ICAL}}, pages 286--298, 2015.
\newblock \href {https://doi.org/10.1007/978-3-662-47666-6_23}
  {\path{doi:10.1007/978-3-662-47666-6_23}}.

\bibitem[KS16]{KlinS16}
Bartek Klin and Micha{\l} Szynwelski.
\newblock {SMT} solving for functional programming over infinite structures.
\newblock In {\em {MSFP}}, volume 207 of {\em {EPTCS}}, pages 57--75, 2016.

\bibitem[KST12]{KurzST12}
Alexander Kurz, Tomoyuki Suzuki, and Emilio Tuosto.
\newblock On nominal regular languages with binders.
\newblock In {\em {FOSSACS}}, pages 255--269, 2012.
\newblock \href {https://doi.org/10.1007/978-3-642-28729-9_17}
  {\path{doi:10.1007/978-3-642-28729-9_17}}.

\bibitem[KT17]{KopczynskiT17}
Eryk Kopczynski and Szymon Toru{\'{n}}czyk.
\newblock {LOIS:} syntax and semantics.
\newblock In {\em {POPL}}, pages 586--598. {ACM}, 2017.

\bibitem[KV94]{KearnsV94}
Michael~J. Kearns and Umesh~V. Vazirani.
\newblock {\em An Introduction to Computational Learning Theory}.
\newblock {MIT} Press, 1994.
\newblock URL:
  \href{https://mitpress.mit.edu/books/introduction-computational-learning-theory}{\texttt{https://mitpress.mit.edu/books/introduction-computational-\\learning-theory}}.

\bibitem[KZ10]{KaminskiZ10}
Michael Kaminski and Daniel Zeitlin.
\newblock Finite-memory automata with non-deterministic reassignment.
\newblock {\em Int. J. Found. Comput. Sci.}, 21(5):741--760, 2010.
\newblock \href {https://doi.org/10.1142/S0129054110007532}
  {\path{doi:10.1142/S0129054110007532}}.

\bibitem[MAMU15]{MyersAMU15}
Robert S.~R. Myers, Jir{\'{\i}} Ad{\'{a}}mek, Stefan Milius, and Henning Urbat.
\newblock Coalgebraic constructions of canonical nondeterministic automata.
\newblock {\em Theor. Comput. Sci.}, 604:81--101, 2015.

\bibitem[Moe19]{Moerman19}
Joshua Moerman.
\newblock {\em Nominal Techniques and Black Box Testing for Automata Learning}.
\newblock PhD thesis, Radboud University, Nijmegen, The Netherlands, 2019.
\newblock URL: \url{http://hdl.handle.net/2066/204194}.

\bibitem[MQ19]{MottetQ19}
Antoine Mottet and Karin Quaas.
\newblock The containment problem for unambiguous register automata.
\newblock In {\em {STACS}}, pages 53:1--53:15, 2019.
\newblock \href {https://doi.org/10.4230/LIPIcs.STACS.2019.53}
  {\path{doi:10.4230/LIPIcs.STACS.2019.53}}.

\bibitem[MS20]{MoermanS20}
Joshua Moerman and Matteo Sammartino.
\newblock Residual nominal automata.
\newblock In Igor Konnov and Laura Kov{\'{a}}cs, editors, {\em {CONCUR}},
  volume 171 of {\em LIPIcs}, pages 44:1--44:21, 2020.
\newblock \href {https://doi.org/10.4230/LIPIcs.CONCUR.2020.44}
  {\path{doi:10.4230/LIPIcs.CONCUR.2020.44}}.

\bibitem[MSS{\etalchar{+}}17]{MoermanS0KS17}
Joshua Moerman, Matteo Sammartino, Alexandra Silva, Bartek Klin, and Micha{\l}
  Szynwelski.
\newblock Learning nominal automata.
\newblock In {\em {POPL}}, pages 613--625. {ACM}, 2017.

\bibitem[Ner58]{Nerode58}
Anil Nerode.
\newblock Linear automaton transformations.
\newblock {\em Proceedings of the AMS}, 9:541– 544, 1958.

\bibitem[NSV04]{NevenSV04}
Frank Neven, Thomas Schwentick, and Victor Vianu.
\newblock Finite state machines for strings over infinite alphabets.
\newblock {\em {ACM} Trans. Comput. Log.}, 5(3):403--435, 2004.

\bibitem[Pit13]{Pitts13}
Andrew~M. Pitts.
\newblock {\em Nominal sets: Names and symmetry in computer science}.
\newblock Cambridge Tracts in Theoretical Computer Science. Cambridge
  University Press, 2013.

\bibitem[SKMW17]{SchroderKMW17}
Lutz Schr{\"{o}}der, Dexter Kozen, Stefan Milius, and Thorsten Wi{\ss}mann.
\newblock Nominal automata with name binding.
\newblock In {\em {FOSSACS}}, pages 124--142, 2017.
\newblock \href {https://doi.org/10.1007/978-3-662-54458-7_8}
  {\path{doi:10.1007/978-3-662-54458-7_8}}.

\bibitem[US20]{UrbatS20}
Henning Urbat and Lutz Schr{\"{o}}der.
\newblock Automata learning: An algebraic approach.
\newblock In {\em {LICS}}, pages 900--914. {ACM}, 2020.

\bibitem[Vaa17]{Vaandrager17}
Frits~W. Vaandrager.
\newblock Model learning.
\newblock {\em Commun. {ACM}}, 60(2):86--95, 2017.

\end{thebibliography}

\end{document}